\documentclass[10pt,journal,twocolumn]{IEEEtran}
\usepackage{cite}
\usepackage{color}
\usepackage{algorithmic}
\usepackage{epsfig}
\usepackage{amssymb}
\usepackage[tbtags]{amsmath}
\usepackage{graphics,eepic,epic}
\usepackage{latexsym}
\usepackage{euscript}
\usepackage{subfigure}
\usepackage{graphics,eepic,epic,psfrag}
\usepackage{algorithm}
\usepackage{algorithmic}
\usepackage{amsthm}
\usepackage[font=footnotesize]{caption}
\usepackage{xcolor}
\usepackage{etoolbox}
\newtheorem{proposition}{Proposition}

\newcommand{\edww}[1]{{\color{black}#1}}
\newcommand{\edwER}[1]{{\color{black}#1}}
\newcommand{\edz}[1]{{\color{black}#1}}
\newcommand{\edf}[1]{{\color{black}#1}}
\newcommand{\edwn}[1]{{\color{black}#1}}
\newcommand{\edff}[1]{{\color{black}#1}}
\newcommand{\edzr}[1]{{\color{black}#1}}
\newcommand{\edzs}[1]{{\color{black}#1}}

\newcommand{\rr}[1]{{\color{black}#1}}
\makeatletter
\pretocmd\@bibitem{\color{black}\csname keycolor#1\endcsname}{}{\fail}
\newcommand\citecolor[1]{\@namedef{keycolor#1}{\color{blue}}}
\makeatother
\begin{document}
\title{Multi-Cell Sparse Activity Detection for Massive \edzr{Random Access}: Massive MIMO versus Cooperative MIMO}

\author{Zhilin~Chen,~\IEEEmembership{Student Member,~IEEE,}
        Foad~Sohrabi,~\IEEEmembership{Member,~IEEE,}
        and~Wei~Yu,~\IEEEmembership{Fellow,~IEEE}
\thanks{The authors are with The Edward S. Rogers Sr. Department of Electrical
and Computer Engineering, University of Toronto, Toronto, ON M5S 3G4,
Canada (e-mails:\{zchen, fsohrabi, weiyu\}@comm.utoronto.ca). This work has been presented in part at IEEE International Conference on Acoustics, Speech, and Signal Processing (ICASSP), Calgary, Canada, April 2018 \cite{Chen2018C}. \edff{This work is supported by the Natural Sciences and Engineering Research Council of Canada.
}}}

\maketitle

\begin{abstract}
This paper considers sparse device activity detection for cellular machine-type communications with non-orthogonal signatures using the approximate message passing algorithm. This paper compares two network architectures, massive multiple-input multiple-output (MIMO) and cooperative MIMO, in terms of their effectiveness in overcoming inter-cell interference. In the massive MIMO architecture, each base station (BS) detects only the users from its own cell while treating inter-cell interference as noise. In the cooperative MIMO architecture, each BS detects the users from neighboring cells as well; the detection results are then forwarded in the form of log-likelihood ratio (LLR) to a central unit where final decisions are made. This paper analytically characterizes the probabilities of false alarm and missed detection for both architectures. Numerical results validate the analytic characterization and show that as the number of antennas increases, a massive MIMO system effectively drives the detection error to zero, while as the cooperation size increases, the cooperative MIMO architecture mainly improves the cell-edge user performance. Moreover, this paper studies the effect of LLR quantization to account for the finite-capacity fronthaul. Numerical simulations of a practical scenario suggest that in that specific case cooperating three BSs in a cooperative MIMO system achieves about the same cell-edge detection reliability as a non-cooperative massive MIMO system with four times the number of antennas per BS.
\end{abstract}


\section{Introduction}
Massive machine-type communications (mMTC) aim to meet the demand for wireless connectivity for tens of millions of devices with event-driven traffic in application domains of the future fifth generation (5G) cellular infrastructure \cite{Durisi2016}. A main challenge of mMTC is scalable and efficient random access design in the uplink for a large pool of devices, among which only a small fraction of devices are active in each coherence time interval due to sporadic traffic \cite{Shirvanimoghaddam2017}. Under such circumstances, it is necessary for base stations (BSs) to identify the active devices as an initial step to enable subsequent data receiving and decoding processes \cite{Wei2016}. This paper studies the sporadic user activity detection for the 5G wireless cellular networks.

This paper considers a pilot-based random access protocol, where the active users transmit their unique signatures synchronously to their BSs at the start of a random access procedure. To accommodate a large number of devices within limited coherence time interval, non-orthogonal signature sequences are employed. \edzr{This differs from the conventional random multiple access in which orthogonal sequences are used and collisions are handled via retransmission. With non-orthogonal sequences, the BSs perform user activity detection based on the received signal, which is a combination of the signatures of active users.} Considering the sparse nature of the activity patterns of the devices and the need for scalable low-complexity algorithm for activity detection, this paper adopts the approximate message passing (AMP) algorithm from compressed sensing.

In conventional cellular networks, inter-cell interference is often seen as a crucial limiting factor, as each BS operates independently by treating the interference as background noise. As a result, the user activity detection accuracy would be \edzr{severely} impacted by inter-cell interference in a multi-cell system as compared to a single-cell system. This paper investigates two promising network architectures for alleviating the inter-cell interference: \edzr{non-cooperative} massive multiple-input and multiple-output (MIMO) and cooperative MIMO, and compares their effectiveness in combating
\edww{inter-cell interference}. Although these two architectures have been extensively studied in the literature from the perspective of
maximizing the achievable rate for data transmission, e.g., \cite{Jungnickel2014,Hosseini2014a}, and references therein, the comparison in terms of user activity detection is not yet available. The main goal of this paper is to characterize the performance of user activity detection for both the massive MIMO and the cooperative MIMO architectures in multi-cell systems.

For massive MIMO networks \edzr{without BS cooperation}, each BS operates independently---only focusing on detecting the active users from its own cell while treating the inter-cell interference as noise, \edww{because} without cooperation, the BSs are unlikely to obtain knowledge about the out-of-cell users. \edww{Nevertheless}, the impact of interference on user activity detection can be \edww{significantly} alleviated by increasing the number of antennas per BS, as shown in this paper.

For cooperative MIMO networks \edzr{where the antenna arrays are not necessarily large}, the BSs are connected in a cloud-radio access network (C-RAN) architecture,
\edww{thus} the knowledge of all the users in the network can be shared among the BSs, based on which the inter-cell interference can be exploited rather than treated as noise.
\edww{In such a network, depending on the functional split of C-RAN, each BS can either forward the received signal to a central unit (CU) for centralized user activity detection, or perform a preliminary user activity detection and forward the preliminary detection results to the CU for final decisions.}
This paper focuses on the detecting and forwarding approach, which is more practical, while taking
complexity and fronthaul capacity requirements into consideration. Note that in such an approach, each BS should detect \edww{not only the users} from its own cell but also from \edww{the neighboring cells, so} that multiple detection results with respect to the same user can be obtained \edww{from} different BSs and aggregated at the CU.

The main objective of this paper is to quantify the benefits of massive MIMO versus cooperative MIMO in terms of two types of detection error: the probability of false alarm and the probability of missed detection.
This paper shows that massive MIMO can substantially improve the performance of all users, whereas cooperative MIMO mainly benefits the cell-edge users. Further, this paper shows that
\edww{significantly} fewer antennas per BS are needed for cooperative MIMO as compared to massive MIMO to achieve comparable performance for
the cell-edge users.

\subsection{Related Work}
Device activity detection problem has been investigated in a variety of wireless systems using different approaches. For example, \cite{Dekorsy2015,Wunder2015} propose the use of compressed sensing techniques for joint user activity detection and data detection/channel estimation in cellular systems
without considering the effect of inter-cell interference. In code-division multiple access
systems, \edww{sparse user activity detection} is considered jointly with multi-user detection via a sparsity-exploiting
maximum \emph{a posteriori}
approach in \cite{ZhuGiannakis2011}. By further exploiting channel statistics, \cite{Chen2018} adopts the AMP algorithm with Bayesian denoiser for activity detection, and characterizes the detection performance. In \cite{Boljanovic2017}, two approaches, compressed sensing technique and coded slotted ALOHA, are compared in terms of detection accuracy and energy efficiency in user activity detection.

In the context of
massive MIMO systems, the user activity detection is considered in \cite{Liu2018,Bjornson2017,Haghighatshoar2018,Senel2017}. Assuming non-orthogonal Gaussian sequences, \cite{Liu2018} studies the user activity detection performance in single-cell scenario in the asymptotic regime, showing that perfect detection can be achieved by employing AMP. By using mutually orthogonal
pilot sequences and designing an uncoordinated pilot collision resolution protocol, \cite{Bjornson2017} investigates the user activity detection in multi-cell massive MIMO systems and analyzes the collision probability. In the massive MIMO setup, \cite{Haghighatshoar2018} studies the scaling law of user activity detection with finite-length signature sequences by focusing on the covariance matrix
\edww{of the received signal across the antenna domain.}
In particular, \cite{Haghighatshoar2018} shows that using the covariance-based techniques
enables a massive connectivity network to accommodate more users as compared to employing the existing compressed sensing techniques. By embedding one bit information
 in the user activity detection for control signaling,
 \cite{Senel2017} proposes an AMP-based joint user activity detection and information decoding method for massive MIMO systems.

The user activity detection is also studied in C-RAN in \cite{Lau2015,Simeone2016}. A
Bayesian compressed sensing algorithm is proposed in \cite{Lau2015}, where the received signals from all the BSs are concatenated at the CU followed by a joint user activity detection. By
considering limited capacity of the fronthaul links between the BSs and the CU, \cite{Simeone2016} compares two schemes---centralized detection with received signal quantization and distributed detection with log-likelihood ratio \edz{(LLR)} quantization---via simulations, demonstrating that centralized detection is preferred with high fronthaul capacity whereas distributed detection is preferred with low fronthaul capacity.

Besides the aforementioned works on design and analysis for practical networks, there are also related works \cite{ChenGuo2017,Wei2016,Polyanskiy2017,Or2017} that address the massive random access problem with sparse user activity from information theoretical perspectives.

\subsection{Main Contributions}
This paper studies the user activity detection problem with non-orthogonal signature sequences for massive connectivity in cellular networks \edww{using the AMP algorithm.} Two potential network architectures, massive MIMO and cooperative MIMO, are investigated and compared. The main contributions of this paper are summarized as follows.

For the massive MIMO architecture, this paper employs the AMP algorithm for the user activity detection by treating the inter-cell interference as noise. Through a state evolution analysis in the asymptotic regime where the number of users per cell and the length of signature sequences tend to infinity, while keeping their ratio fixed, this paper investigates the impact of the inter-cell interference, and reveals the relation \edzr{between the multi-cell system parameters and the performance of AMP, based on which the probability of false alarm and the probability of missed detection are characterized. By further letting the number of antennas per BS go to infinity, the asymptotic results show that the probability of false alarm and the probability of missed detection can be effectively driven to zero, even
in the existence of inter-cell interference.}

For the cooperative MIMO architecture, this paper considers a cooperative detection scheme, where each BS seeks to detect the users from its own cell as well as several neighboring cells by recovering the inter-cell interference, with the help of the knowledge of the signature sequences and the channel \edzr{statistics} of \edww{the} out-of-cell users. The detection results are then forwarded in the form of \edz{LLR} to the CU, where final decisions on the user activities are made based on an aggregation of the results from all the BSs. \edzr{Through a state evolution analysis, this paper shows that as compared to the massive MIMO case where the inter-cell interference is simply treated as noise, recovering the inter-cell interference at each BS by detecting out-of-cell users in cooperative MIMO achieves better performance.}
This paper \edww{first} characterizes \edf{the probabilities of false alarm and missed detection} by assuming \edww{infinite-capacity} fronthaul links, then further considers the impact of the limited capacity of the fronthaul links by proposing a quantization design that takes the variation of the dynamic range of LLRs for different users into consideration.

This paper conducts extensive simulation studies to validate the analytical results and to compare massive MIMO and cooperative MIMO architectures for user activity detection. The results confirm the analytical characterization of the detection performance for both architectures, and show that massive MIMO is effective in improving the performance of all users, while cooperative MIMO is \edww{most} effective in improving the performance of the cell-edge users. \edzs{In terms of cell-edge user performance, based on numerical simulations of a practical scenario, cooperating three BSs in a cooperative MIMO system achieves roughly the same performance as having four times the number of antennas per BS in a non-cooperative massive MIMO system in that specific case.}

\edww{By further considering the capacity limits of the fronthaul links, we observe that for cooperative MIMO with each user cooperatively detected by three BSs, the performance of the proposed quantization scheme with about 3-4 quantization bits per LLR can already approach the performance of the infinite-capacity fronthaul case.}

\subsection{Paper Organization and Notations}
The remainder of the paper is organized as follows. Section~\ref{sec.sysmod} describes the multi-cell system model and introduces the sparse user activity detection problem for massive MIMO and cooperative MIMO architectures. Section~\ref{sec.amp} presents the AMP-based user activity detection algorithms for both architectures, \edzr{and their performances are analyzed in Section~\ref{sec.massive}.}
Section~\ref{sec.simu} presents the simulation results. \edww{Section~\ref{sec.con} concludes the paper.}

Throughout this paper, upper-case and lower-case letters denote random variables and their realizations, respectively. Boldface lower-case letters denote vectors. Boldface upper-case letters denote matrices or random vectors, where context should make the distinction clear. Superscripts $(\cdot)^{T}$ and $(\cdot)^{*}$
denote transpose and conjugate transpose,
respectively. Further, $\mathbf{I}$ denotes identity matrix with appropriate dimensions, $\mathbb{E}[\cdot]$ denotes expectation operation, $\triangleq$ denotes definition, $|\cdot|$ denotes the cardinality of a set, and $\|\cdot\|_2$ denotes the $\ell_2$ norm.

\section{System Model and Sparse Activity Detection in Two \edf{Architectures}}
\label{sec.sysmod}
\subsection{System Model}
Consider a wireless cellular network comprising $B$ cells indexed by $1,2,\cdots,B$. Each cell contains one BS equipped with $M$ antennas at the center, serving $N$
\edww{uniformly}
distributed single-antenna users. Due to the sporadic traffic of mMTC, only a small subset of total $BN$ users in the network are active in each coherence time interval. Let $a_{bn}\in\{1,0\}$ indicate whether or not user $n$ in cell $b$ is active. We statistically model $a_{bn}$ as independent and identically distributed (i.i.d.) Bernoulli random variables with $\mathrm{Pr}(a_{bn}=1)=\lambda$ where $\lambda$ is a small constant. For the purpose of user identification and channel estimation in random access procedure, each user is assigned a unique length-$L$ signature sequence $\mathbf{s}_{bn}=[s_{bn1},
s_{bn2},\cdots,s_{bnL}]\in \mathbb{C}^{1\times L}$. Assuming that the channel is static in each coherence time interval, and all users transmit their signature sequences with the same power, the received signal $\mathbf{Y}_b \in \mathbb{C}^{L\times M}$ at BS $b$ is
\begin{align}\label{eq.sys}
\mathbf{Y}_b&=\sum_{n=1}^{N}a_{bn}\mathbf{s}_{bn}^{T}\mathbf{h}_{bbn}+\sum_{j\neq b}^{}\sum_{n=1}^{N}a_{jn}\mathbf{s}_{jn}^{T}\mathbf{h}_{bjn}+\mathbf{W}_b\nonumber\\
&=\mathbf{S}_b\mathbf{X}_{bb}+\sum_{j\neq b}^{}\mathbf{S}_{j}\mathbf{X}_{bj}+\mathbf{W}_b,
\end{align}
where $\mathbf{h}_{bjn}\in \mathbb{C}^{1\times M}$ is the channel from user $n$ in cell $j$ to BS $b$, $\mathbf{W}_b\in \mathbb{C}^{M\times L}$ is the effective i.i.d.\ Gaussian noise whose variance $\sigma_{w}^2$ depends on the signal-to-noise ratio (SNR) at the BS, $\mathbf{S}_j\triangleq[\mathbf{s}_{j1}^{T},\cdots,\mathbf{s}_{jN}^{T}]\in\mathbb{C}^{L\times N}$ comprises all the sequences of users in cell $j$, and $\mathbf{X}_{bj}\triangleq[\mathbf{x}_{bj1}^{T},\cdots,\mathbf{x}_{bjN}^{T}]^{T}\in\mathbb{C}^{N\times
M}$, where $\mathbf{x}_{bjn}\triangleq a_{jn}\mathbf{h}_{bjn}\in\mathbb{C}^{1\times M}$ is the row vector of $\mathbf{X}_{bj}$. The second term corresponds to the signal from outside of cell $b$.

For massive \edzs{random access applications}, we are interested in the regime where the number of potential users per cell is much larger than the length of signature sequences, i.e., $N \gg L$. Due to the insufficient sequence dimensions, we cannot assign mutually orthogonal sequences to each potential user. Instead of assuming a carefully designed set of signature sequences or any specific sequence reuse scheme among cells, this paper assumes that the signature sequences used in the entire network are generated according to i.i.d.\ complex Gaussian distribution with zero
mean and variance $1/L$, i.e., $\mathbf{s}_{jn} \thicksim \mathcal{CN}(0,1/L), \forall j, n,$, such that each sequence is unique with unit power.

The inter-cell interference brought by
non-orthogonal signature sequences severely affects the activity detection accuracy. This paper studies two
network architectures, massive MIMO and cooperative MIMO, and compares their effectiveness in overcoming  the inter-cell interference for user activity detection.

\subsection{Sparse Activity Detection in Massive MIMO}
For massive MIMO, we assume that each BS operates independently
to detect only the active users from its own cell.
In this case, inter-cell interference has to be treated as noise, because unless BSs cooperate,
it is unlikely that each BS can obtain knowledge about the out-of-cell users,
which is necessary for performing interference detection
\edww{and cancellation}.
\edww{Nevertheless,}
we can deploy a large number of antennas at each BS to combat the interference and to improve the reliability of user activity detection.

With inter-cell interference treated as noise, we re-write the received signal in \eqref{eq.sys} as
\begin{align}\label{eq.model.tin}
\mathbf{Y}_b=\mathbf{S}_b\mathbf{X}_{bb}+\mathbf{\tilde{W}}_b,
\end{align}
where
\edww{the combined noise and interference}
$\mathbf{\tilde{W}}_b\triangleq \sum_{j\neq b}\mathbf{S}_{j}\mathbf{X}_{bj}+\mathbf{W}_b$
\edww{is}
approximated as Gaussian noise with covariance matrix $\tilde{\sigma}_{w}^2\mathbf{I}$ \edzr{because the the inter-cell interference consists of a large number of independent signals, which tends to the Gaussian distribution as suggested in \cite{Andrews2007}. This Gaussian approximation helps simplify the signal model and the performance analysis. We characterize the value of $\tilde{\sigma}_{w}^2$ by exploiting the statistics of the channels and the signature sequences when analyzing the impact of the inter-cell interference.}

We aim to detect the active users by identifying the non-zero rows in $\mathbf{X}_{bb}$ based on $\mathbf{Y}_b$, which corresponds to solving a compressed sensing problem with multiple measurement vectors (MMV).
\edww{To ensure}
the scalability of the problem and the tractability of the detection performance
\edww{analysis},
this paper employs the computationally efficient AMP algorithm. Note that AMP is also able to exploit the statistics of $\mathbf{X}_{bb}$. To obtain a statistical model of $\mathbf{X}_{bb}$, we model the channel as $\mathbf{h}_{bbn}=g_{bbn}\mathbf{\bar{h}}_{bbn}$, where $g_{bbn}$ is the large-scale fading coefficient assumed to be available at BS $b$, and $\mathbf{\bar{h}}_{bbn}$ is the Rayleigh fading component following $\mathcal{CN}(0,\mathbf{I})$ \edwn{uncorrelated across the antennas and} not known at the BS. \edzr{The $n$-th row of $\mathbf{X}_{bb}$ then can be written as $\mathbf{x}_{bbn}=a_{bn}g_{bbn}\mathbf{\bar{h}}_{bbn}$. By noting that $a_{bn}$ follows the Bernoulli distribution with $\mathrm{Pr}(a_{bn}=1)=\lambda$
in an i.i.d.\ fashion, $\mathbf{x}_{bbn}$ can be statistically modeled by the following mixed Bernoulli-Gaussian distribution}
\begin{align}\label{eq.distri.bg}
\mathbf{x}_{bbn} \thicksim (1-\lambda)\delta_{\mathbf{0}}+\lambda\mathcal{CN}(0,g_{bbn}^2\mathbf{I}),
\end{align}
where the Gaussian distribution is parameterized by $g_{bbn}$ and $\delta_{\mathbf{0}}$ is a point mass at $\mathbf{0}$.

\subsection{Sparse Activity Detection in Cooperative MIMO}
For cooperative MIMO, we assume an additional CU deployed in the network with fronthaul connections to all $B$ BSs, forming a C-RAN system. There are different ways to exploit C-RAN to enable cooperation. \edzr{For example, the authors in \cite{Simeone2016} consider two strategies:} \edzs{a C-RAN architecture with} \edzr{functional split, in which a preliminary detection is performed at each BS and the detection results are forwarded to the CU, where final decisions on user activity are carried out;} \edzs{and a centralized C-RAN architecture with} \edzr{centralized cooperation, in which all the received signals are collected and processed at the CU}. This paper adopts a functional split strategy. This paper does not consider centralized cooperation due to its higher requirements on
\edww{the}
computation capability
\edww{at the CU
and the capacity of the fronthaul}.

\subsubsection{Preliminary User Activity Detection at Each BS}
Under the C-RAN architecture, it is likely that each BS can obtain some information about the out-of-cell users, e.g., their signature sequences and channel statistics, which allow the BS to exploit the inter-cell interference to detect the active users not only from its own cell, but also from several neighboring cells. As an example, we consider a special case where each BS seeks to detect all active users in the network. The received signal in \eqref{eq.sys} can be re-written as
\begin{align}\label{eq.sys2}
\mathbf{Y}_b &= \left[\begin{array}{ccc}\mathbf{S}_1,&\cdots,&\mathbf{S}_B\end{array}\right]\left[\begin{array}{c}\mathbf{X}_{b1}\\\vdots\\\mathbf{X}_{bB}\end{array}\right]+\mathbf{W}_b\nonumber\\
&\triangleq \mathbf{S}\edwER{\mathbf{X}_b}+\mathbf{W}_b, 
\end{align}
which is
a \edww{\emph{recovering-inter-cell-interference}} strategy rather than a \edww{\emph{treating-interference-as-noise}} strategy.
This paper analyzes these two strategies
and shows that recovering inter-cell interference brings improvement as compared to treating interference as noise.

Note that \eqref{eq.sys2} also corresponds to a compressed sensing problem, thus still allowing the use of the AMP algorithm. To exploit the channel statistics, we model the channel from any user in the network to BS $b$ as
$\mathbf{h}_{bjn}=g_{bjn}\mathbf{\bar{h}}_{bjn}$, where the  value of the large-scale fading coefficient $g_{bjn}$ is assumed to be available at BS $b$. Each row of $\mathbf{X}_{bj}$ follows a distribution similar to \eqref{eq.distri.bg} as
\begin{align}\label{eq.distri.bg.2}
\mathbf{x}_{bjn} \thicksim (1-\lambda)\delta_{\mathbf{0}}+\lambda\mathcal{CN}(0,g_{bjn}^2\mathbf{I}),
\end{align}

\subsubsection{Final Decision at the CU}
Based on the preliminary user activity detection phase, each BS makes soft decision on the activities of all users in the network in the form of LLR, and forwards those LLRs to the CU, where final decisions on the user activities are performed via LLR aggregation.
\edww{Note that due to the
constraint on the fronthaul capacity
in practice the LLR values need to be quantized.}

In this paper, we design the AMP-based user activity detection algorithms for the massive MIMO system as well as the cooperative MIMO system, and analytically characterize the detection error for both scenarios.
\edww{The analytic results allow an efficient performance comparison of the two cases via numerical methods.}

\section{AMP-based Activity Detection for Massive MIMO and Cooperative MIMO}
\label{sec.amp}
This section presents the AMP-based user activity detection algorithms for massive MIMO and cooperative MIMO. Since the user activity detection in both scenarios involves solving a sparse recovery problem at each BS, we first discuss the AMP algorithm
\edww{for compressed sensing.}
After an estimate of $\mathbf{X}_{bb}$ or
\edwER{$\mathbf{X}_b$}
is obtained, we discuss how to perform user activity detection based on
\edww{the}
LLRs in the massive MIMO case, and how to perform LLR forwarding and LLR aggregation for final decisions at the CU in the cooperative MIMO case, respectively.

\subsection{AMP with Bayesian Denoiser}
AMP is an iterative algorithm originally proposed in \cite{Donoho2009} and has been extended
\edww{in \cite{Rangan2011, Kim2011, Maleki2013, Ziniel2013} for various types of}
sparse recovery problems. In this paper, we adopt a variant of AMP used in \cite{Kim2011,Chen2018,Liu2018}.
\edz{In this section, we first consider the use of AMP for massive MIMO, then discuss the AMP algorithm for cooperative MIMO.}

To recover the row sparse matrix $\mathbf{X}_{bb}$ \edz{in \eqref{eq.model.tin} for massive MIMO}, AMP starts with $\mathbf{X}^0=\mathbf{0}$ and $\mathbf{Z}^0=\mathbf{Y}_b$, and proceeds
\edww{in}
each iteration as
\begin{align}
\mathbf{X}^{t+1}&=\eta_t(\mathbf{S}_b^{*}\mathbf{Z}^{t}+\mathbf{X}^{t}),\label{eq.vamp1}\\
\mathbf{Z}^{t+1} &= \mathbf{Y}_b-\mathbf{S}_b\mathbf{X}^{t+1}+ \frac{N}{L}\mathbf{Z}^{t}\langle\eta^{\prime}_t(\mathbf{S}_b^{*}\mathbf{Z}^{t}+\mathbf{X}^{t})\rangle, \label{eq.vamp2}
\end{align}
where $t=0,1,\cdots$ is the iteration index, $\mathbf{X}^{t}$ is
the estimate of $\mathbf{X}_{bb}$ at iteration $t$, $\mathbf{Z}^{t}$ is the residual, $\eta_t(\cdot)\triangleq [\eta_{t}(\cdot,g_{bb1}),\cdots,\eta_{t}(\cdot,g_{bbN})]^{T}$, with $\eta_{t}(\cdot,g_{bbn}): \mathbb{C}^{1\times M}\rightarrow \mathbb{C}^{1\times M}$ being an appropriately designed non-linear function known as \emph{denoiser} that operates on the $n$th row of $\mathbf{S}_b^{*}\mathbf{Z}^{t}+\mathbf{X}^{t}$, with $g_{bbn}$ as the parameter, $\eta^{\prime}_{t}(\cdot)\triangleq[\eta_{t}^{\prime}(\cdot,g_{bb1}),\cdots,\eta_{t}^{\prime}(\cdot,g_{bbN})]^{T}$, with $\eta_{t}^{\prime}(\cdot,g_{bbn})$ being the first order derivative of $\eta_{t}(\cdot,g_n)$, and $\langle\cdot\rangle$ is the \edzr{sample mean} of all $N$ derivatives.

A useful property of AMP is that the matched filtered output
$\mathbf{\tilde{X}}^{t}\triangleq\mathbf{S}_b^{*}\mathbf{Z}^{t}+\mathbf{X}^{t}$
in \eqref{eq.vamp1}
can be statistically modeled as
\edww{a linear model of the
signal itself \edff{plus} the
noise-plus-multiuser-interference term}, i.e., $\mathbf{\tilde{X}}^{t}=\mathbf{X}_{bb}+\mathbf{V}^{t}$,
where the row vectors of $\mathbf{V}^{t}$ are Gaussian with covariance matrix $\mathbf{\Sigma}_{t}$, which can be tracked in the asymptotic regime where $L,N\rightarrow\infty$, with their ratio $L/N$ fixed, via the state evolution. \edzr{To describe the state evolution conveniently, we first introduce a few random variables $\mathbf{R}\in\mathbb{C}^{1\times M}$, $\mathbf{U}^t\in\mathbb{C}^{1\times M}$ and $G_b\in\mathbb{R}$, where $\mathbf{R}$ follows the same distribution as the row vectors of $\mathbf{X}_{bb}$, $\mathbf{U}^t$ follow the Gaussian distribution with zero mean and covariance matrix $\mathbf{\Sigma}_{t}$, and $G_b$ follows the distribution of the large-scale fading coefficient $g_{bbn}, \forall n$} \edzs{assuming that the} \edzr{users in cell $b$ are} \edzs{uniformly} \edzr{and independently located} \edzs{in the cell}. \edzr{The state evolution can} \edzs{then} \edzr{be expressed as}
\begin{align}\label{eq.se}
\mathbf{\Sigma}_{t+1} = \tilde{\sigma}_{w}^2\mathbf{I}+\frac{N}{L}\mathbb{E}\left[\mathbf{D}^{t}(\mathbf{D}^{t})^{*}\right],
\end{align}
\edzr{where $\mathbf{D}^{t}\triangleq \left(\eta_{t}(\mathbf{R}+\mathbf{U}^{t},G_b)-\mathbf{R}\right)^{T}\in \mathbb{C}^{M\times 1}$, and the expectation is taken with respect to all random variables $\mathbf{R}$, $\mathbf{U}^{t}$, and $G_b$.}

Since $\mathbf{\Sigma}_{t+1}$ contains the statistical information on the noise-plus-multiuser-interference term $\mathbf{V}^{t+1}$, we can characterize the performance of AMP from $\mathbf{V}^{t+1}$, whose limiting value as $t\rightarrow \infty$ is given by the fixed point of the equation \eqref{eq.se}. Moreover, it
is shown in \cite{Chen2018,Liu2018} that with the
\edwn{statistical} model of $\mathbf{X}_{bb}$ in \eqref{eq.distri.bg}, where the row vectors of $\mathbf{X}_{bb}$ are drawn from \edwn{an i.i.d.\ }Bernoulli-Gaussian distribution \edwn{and uncorrelated across the antennas}, $\mathbf{\Sigma}_{t+1}$ stays as a diagonal matrix with identical diagonal entries at each iteration, i.e., $\mathbf{\Sigma}_{t+1}=\tau_{t+1}^2\mathbf{I}$, \edzr{which simplifies the state evolution to a one-dimension equation as
$\tau_{t+1}^2\mathbf{I} = \tilde{\sigma}_{w}^2\mathbf{I}+NL^{-1}\mathbb{E}\left[\mathbf{D}^{t}(\mathbf{D}^{t})^{*}\right]$ that
can be easily tracked.}

Now we discuss the design of $\eta_{t}(\cdot,g_{bbn})$ by exploiting the statistics of $\mathbf{X}_{bb}$ in \eqref{eq.distri.bg}, and the Gaussianity of $\mathbf{V}^{t}$. Recall that the $n$th row vector of the matched filtered output $\mathbf{\tilde{X}}^{t}\triangleq\mathbf{S}_b^{*}\mathbf{Z}^{t}+\mathbf{X}^{t}$ in \eqref{eq.vamp1} can be modeled as $\mathbf{\tilde{x}}_{bbn}=\mathbf{x}_{bbn} + \mathbf{v}_{n}^{t}$, where $\mathbf{v}_{n}^{t}$ is drawn from $\mathcal{CN}(0,\tau_t^2\mathbf{I})$. The Bayesian denoiser $\eta_{t}(\cdot,g_{bbn})$ operating on $\mathbf{\tilde{x}}_{bbn}$ is \edzr{then designed as the conditional expectation of $\mathbf{x}_{bbn}$ given $\mathbf{\tilde{x}}_{bbn}$} as follows
\cite{Kim2011}
\begin{align}\label{eq.denoiser}
\eta_{t}(\mathbf{\tilde{x}}_{bbn}^{t},g_{bbn})&=\frac{\theta_{bbn}(1+\theta_{bbn})^{-1}\mathbf{\tilde{x}}_{bbn}^{t}}{1+\frac{1-\lambda}{\lambda}(1+\theta_{bbn})^{M}\exp(-\Delta_{bbn}\|\mathbf{\tilde{x}}_{bbn}^{t}\|_2^2)},
\end{align}
where $\theta_{bbn}\triangleq g_{bbn}^2\tau_t^{-2}$, and $\Delta_{bbn}\triangleq \tau_t^{-2}-(g_{bbn}^2+\tau_t^2)^{-1}$.  \edzr{It is worth noting that \eqref{eq.denoiser} is obtained based on the assumption of uncorrelated Rayleigh fading channels. In the case of uncorrelated channel model or other more sophisticated channel models, the approach of AMP is applicable with a re-designed $\eta_{t}(\cdot,g_{bbn})$ to incorporate the new channel statistics.}

The algorithm discussed above can be extended to solve the problem \eqref{eq.sys2} by simply replacing $\mathbf{X}_{bb}$, $\mathbf{S}_b$, $\eta_{t}(\cdot,g_{bbn})$, $\mathbf{x}_{bbn}$, and $\mathbf{\tilde{x}}_{bbn}$ with $\mathbf{X}_{b}$, $\mathbf{S}$, $\eta_{t}(\cdot,g_{bjn})$, $\mathbf{x}_{bjn}$, and $\mathbf{\tilde{x}}_{bjn}$, respectively.
It is worth noting that due to the different way to deal with the inter-cell interference as well as the increase in dimensions in \eqref{eq.sys2}, the state evolution in the cooperative MIMO case \edww{\eqref{eq.sys2}} is
\begin{align}\label{eq.se.2}
\mathbf{\Sigma}_{t+1} = \sigma_{w}^2\mathbf{I}+ \frac{NB}{L}\mathbb{E}\left[\mathbf{D}^{t}(\mathbf{D}^{t})^{*}\right],
\end{align}
where, with slightly abuse of notation, $\mathbf{D}^{t}\triangleq \left(\eta_{t}(\mathbf{R}+\mathbf{U}^{t},G)-\mathbf{R}\right)^{T}$, \edzr{with random variable $G$ following the distribution of the large-scale fading coefficient $g_{bjn}, \forall j, n$ assuming that the users in the entire network are uniformly and independently located in all the cells. This differs from \eqref{eq.se} where we use $G_b$ to statistically model the large-scale fading coefficients of only the users within cell $b$.} A second difference is that $\sigma_{w}^2$ in \eqref{eq.se.2} represents the strength of only the background noise, while $\tilde{\sigma}_{w}^2$ in \eqref{eq.se} represents the strength of the background noise plus the inter-cell interference.

\subsection{User Activity Detection in Massive MIMO}
We
\edww{now apply the above analysis of AMP to the user activity detection problem in massive MIMO case.}
After AMP converges, we adopt LLR test to decide on the activity of each user.
In the massive MIMO system, each BS only seeks to detect the active users from its own cell. The LLR of user $n$ from cell $b$ is
computed at BS $b$ based on the matched filtering output $\mathbf{\tilde{x}}_{bbn}^{t}$. If user $n$ from cell $b$ is inactive, i.e., $a_{bn}=0$,
the likelihood of observing $\mathbf{\tilde{x}}_{bbn}^{t}$ at iteration $t$ is
\begin{align}\label{eq.likeli1}
p(\mathbf{\tilde{x}}_{bbn}^{t}|a_{bn}=0)=\frac{\exp\left(-\|\mathbf{\tilde{x}}_{bbn}^{t}\|_2^2\tau_{t}^{-2}\right)}{\pi^M \tau_{t}^{2M}},
\end{align}
\edf{where we use
$\mathbf{\tilde{x}}_{bbn}=\mathbf{x}_{bbn} + \mathbf{v}_{n}^{t}$ in which
$\mathbf{x}_{bbn}$ is Bernoulli-Gaussian and $\mathbf{v}_{n}^{t}$ is Gaussian.}
In the case that the user is inactive, i.e., $a_{bn}=1$, the likelihood of observing $\mathbf{\tilde{x}}_{bbn}^{t}$ at iteration $t$ is
\begin{align}\label{eq.likeli2}
p(\mathbf{\tilde{x}}_{bbn}^{t}|a_{bn}=1)=\frac{\exp\left(-\|\mathbf{\tilde{x}}_{bbn}^{t}\|_2^2(\tau_{t}^2+g_{bbn}^2)^{-1}\right)}{\pi^M (\tau_{t}^2+g_{bbn}^2)^{M}}.
\end{align}
The LLR for user $n$ from cell $b$ is obtained at BS $b$ as
\begin{align}\label{eq.llr}
\operatorname{LLR}_{bbn} &= \log\left(\frac{p(\mathbf{\tilde{x}}_{bbn}^{t}|a_{bn}=1)}{p(\mathbf{\tilde{x}}_{bbn}^{t}|a_{bn}=0)}\right)\nonumber\\ &=\|\mathbf{\tilde{x}}^{t}_{bbn}\|_2^2\Delta_{bbn} - M\log(1+\theta_{bbn}),
\end{align}
where
\edww{again}
$\theta_{bbn}= g_{bbn}^2\tau_t^{-2}$, and $\Delta_{bbn}= \tau_t^{-2}-(g_{bbn}^2+\tau_t^2)^{-1}$.
Since the large-scale fading coefficient $g_{bbn}$ is assumed to be known at the BSs and $\tau_{t}$
is a constant parameter determined by the state evolution, it can be seen that the LLR expression in \eqref{eq.llr} is only a function of  $\|\mathbf{\tilde{x}}_{bbn}^{t}\|_2^2$.
Further, by observing that $\operatorname{LLR}_{bbn}$ is monotonic in $\|\mathbf{\tilde{x}}_{bbn}^{t}\|_2^2$, we can set a threshold $l_{bn}$ on $\|\mathbf{\tilde{x}}_{bbn}^{t}\|_2^2$ to make a hard decision on the user activity, i.e., user $n$ from cell $b$ is declared to be active if $\|\mathbf{\tilde{x}}_{bbn}^{t}\|_2^2 > l_{bn}$; otherwise, it is declared to be inactive.

\subsection{User Activity Detection in Cooperative MIMO}
In cooperative MIMO, the user activity is detected based on not only the BS from its own cell, but also the BSs from the neighboring cells. For ease of illustration, we assume that each BS seeks to detect the active users in the entire network, \edz{so that} for user $n$ from cell $b$, the LLRs are obtained at all BSs. Similar to \eqref{eq.llr}, the LLR at BS $j$ can be expressed as
\begin{align}\label{eq.llr.coop}
\operatorname{LLR}_{jbn}=\|\mathbf{\tilde{x}}^{t}_{jbn}\|_2^2\Delta_{jbn} - M\log(1+\theta_{jbn}),
\end{align}
\edf{where $\theta_{jbn}\triangleq g_{jbn}^2\tau_t^{-2}$ and $\Delta_{jbn}\triangleq \tau_t^{-2}-(g_{jbn}^2+\tau_t^2)^{-1}$.
Note that the value of $\tau_{t}$ contained in \eqref{eq.llr.coop} is different from that in \eqref{eq.llr} due to the
\edww{difference in the}
received signal models.}

In the cooperative MIMO system, the final decisions on user activity are carried out at the CU, hence the BSs do not need to make hard decisions. Instead, each BS forwards the LLRs to the CU. We first consider an ideal case where LLRs can be forwarded to the CU perfectly. We deal with the non-ideal case in the subsequent subsection.

To make the LLR forwarding more flexible, we assume that each BS, say BS $b$, is able to select a subset of all the users whose LLRs are forwarded to the CU. This is motivated by the fact that for the users far away from BS $b$, the LLRs obtained by BS $b$ may not be very useful at the CU. We describe the BS-user association from a user-centric perspective, in which each user selects a subset of closest BSs for LLR forwarding.
Specifically, for user $n$ from cell $b$, the set of BSs that forward the LLRs of that user to the CU is denoted as $\mathcal{B}_{bn}\subseteq \{1,\cdots,B\}$ with $B_{bn}=|\mathcal{B}_{bn}|$.

Suppose that the LLRs of user $n$ \edf{in} cell $b$ obtained at $\mathcal{B}_{bn}$ are gathered at the CU.
We treat the inter-cell interference at each BS as independent
\edww{so}
that the collected LLRs can be regarded as independent samples, hence the aggregated LLR is the summation of
\edf{those LLRs}
as
\begin{align}\label{eq.llr.agg}
\operatorname{LLR}_{bn}^{\mathrm{AG}}=\sum_{j\in\mathcal{B}_{bn}}\left(\Delta_{jbn}\|\mathbf{\tilde{x}}_{jbn}\|_2^2
- M\log(1+\theta_{jbn})\right).
\end{align}
Note that $\operatorname{LLR}_{bn}^{\mathrm{AG}}$ is monotonic in $\sum_{j\in\mathcal{B}_{bn}}\Delta_{jbn}\|\mathbf{\tilde{x}}_{jbn}\|_2^2$, which is a weighted sum of $\|\mathbf{\tilde{x}}_{jbn}\|_2^2$, and the weights depend on the large-scale fading coefficients $g_{jbn}, j\in\mathcal{B}_{bn}$ and the parameter $\tau_{t}$. Therefore,
\edww{hard decision}
on user activity at the CU can be performed on $\sum_{j\in\mathcal{B}_{bn}}\Delta_{jbn}\|\mathbf{\tilde{x}}_{jbn}\|_2^2$ with some threshold $l_{bn}^{\mathrm{CO}}$, i.e., user $n$ from cell $b$ is declared to be active if $\sum_{j\in\mathcal{B}_{bn}}\Delta_{jbn}\|\mathbf{\tilde{x}}_{jbn}\|_2^2> l_{bn}^{\mathrm{CO}}$; otherwise, it is declared to be inactive.

\subsection{LLR Quantization Based on User-Specific Codebooks in Cooperative MIMO}
\label{sec.quant}
Due to the limited capacity of fronthaul links in practice, this section considers the
\edww{design of quantization scheme for LLRs} at each BS. Since LLR is a real scalar, let $\mathcal{Q}(\cdot;\mathcal{A})$ denote \edww{a}
quantizer applied to a real scalar with a predefined quantization codebook $\mathcal{A}$. Note that the possible value of $\operatorname{LLR}_{jbn}$ depends on both $g_{jbn}$ and $\tau_{t}$, indicating that the dynamic range of $\operatorname{LLR}_{jbn}$ is user-specific.
\edf{Hence,} a common quantization codebook at BS $j$ for all $\operatorname{LLR}_{jbn}, \forall b,n$, may not properly capture the LLRs of all users when the number of quantization bits is small.

We design user-specific codebooks that
\edww{consider}
the distribution of $g_{jbn}$.
Note that $\Delta_{jbn}$ and $\theta_{jbn}$ in the expression of LLR in \eqref{eq.llr.coop} can be seen as constants, which are assumed to be known at the CU.
\edf{Rather than} \edww{quantizing the LLR itself, we choose to quantize $\|\mathbf{\tilde{x}}_{jbn}\|_2^2$ instead for simplicity. We express the quantizer as}
\begin{align}
y = \mathcal{Q}(\|\mathbf{\tilde{x}}_{jbn}\|_2^2; \mathcal{A}_{jbn}),
\end{align}
where codebook $\mathcal{A}_{jbn}$ is designed by exploiting the statistics of $\|\mathbf{\tilde{x}}_{jbn}\|_2^2$.
\edww{Ideally, since}
$\mathbf{\tilde{x}}_{jbn}=\mathbf{x}_{jbn} + \mathbf{v}_{n}^{t}$, we observe that $\mathbf{\tilde{x}}_{jbn}$ follows a mixture Gaussian distribution, depending on whether user $n$ from cell $b$ is active or not, as
\begin{align}
\mathbf{\tilde{x}}_{jbn} \thicksim (1-\lambda)\mathcal{CN}(0,\tau_{t}^2\mathbf{I})+\lambda\mathcal{CN}(0,(g_{jbn}^2+\tau_{t}^2)\mathbf{I}),
\end{align}
\edww{so the}
Lloyd's algorithm can be employed to generate the codebook for $\|\mathbf{\tilde{x}}_{jbn}\|_2^2$.
\edf{However, this approach requires
\edww{a different}
$\mathcal{A}_{jbn}$
for each
\edww{$(b,n)$,}}
which can be expensive from the viewpoint of computational complexity or storage cost. Instead, we consider a simple uniform quantization scheme, which
\edww{turns out}
to be very effective
\edww{as shown}
in simulations in Section~\ref{sec.simu.quant}.

\edf{We design the}
\edww{uniform}
codebooks for different users with different LLR dynamic ranges as
\begin{align}
\mathcal{A}_{jbn}\triangleq\left\{\frac{l_{max}^{jbn}}{2^{Q+1}},\frac{3l_{max}^{jbn}}{2^{Q+1}},\cdots,
\frac{(2^{Q+1}-1)l_{max}^{jbn}}{2^{Q+1}}\right\},
\end{align}
where $Q$ is the number of quantization bits, $l_{max}^{jbn}$
\edww{is chosen to capture a large}
percentage of the dynamic range of $\|\mathbf{\tilde{x}}_{jbn}\|_2^2$,
\edww{i.e., to satisfy}
\begin{align}
\mathrm{Pr}(\|\mathbf{\tilde{x}}_{jbn}\|_2^2\leq l_{max}^{jbn})=\zeta,
\end{align}
where $\zeta$ is a predefined constant close to one.
\edww{Note that}
the value of $l_{max}^{jbn}$ can be pre-computed accordingly as a look-up table.

\edww{The} value of $\zeta$ should be properly chosen based on the tail probability of $\|\mathbf{\tilde{x}}_{jbn}\|_2^2$\edwn{,}
\edww{especially when $Q$ is small}.
If $\zeta$ is very close to one, the range $[0,l_{max}^{jbn}]$ may become too large which results in quantization levels that cannot capture the small values of $\|\mathbf{\tilde{x}}_{jbn}\|_2^2$ especially when the user is inactive; if $\zeta$ is not sufficiently close to one, the range $[0,l_{max}^{jbn}]$ may not properly cover the possible large values of $\|\mathbf{\tilde{x}}_{jbn}\|_2^2$ when the user is active. In general, larger $\zeta$
\edww{should be chosen if $Q$ is large, or if the}
tail of $\|\mathbf{\tilde{x}}_{jbn}\|_2^2$
\edww{falls to zero quickly.}
It is worth noting that the tail of $\|\mathbf{\tilde{x}}_{jbn}\|_2^2$
\edww{falls more quickly}
when the number of antennas $M$ increases.

Note that the proposed uniform quantization scheme is not necessarily optimal in terms of user activity detection performance. However, simulation results in Section~\ref{sec.simu.quant} reveal that for typical cellular systems with massive devices, such a simple design with three or four quantization bits can already approach the performance of the infinite fronthaul capacity case.

\section{\edzr{Performance Analysis of User Activity Detection in Massive MIMO and Cooperative MIMO}}
\label{sec.massive}
\edf{This section characterizes the probabilities of false detection and missed detection for
user activity detection in massive MIMO and cooperative MIMO.} \edzr{The main tool used in the analysis is the state evolution described in \eqref{eq.se} and \eqref{eq.se.2}. Although the state evolution holds asymptotically in the regime where $L,N\rightarrow\infty$ with fixed ratio ratio $L/N$, it provides accurate performance predictions for large but finite $L$ and $N$, which is the typical setup in mMTC with a large number low-mobility and low-rate devices.

In this section, we first derive explicit expressions for the state evolution in massive MIMO and cooperative MIMO and analytically compare them. We then characterize the detection error probabilities for both cases, based on which we finally discuss the asymptotic properties. }

\subsection{\edzr{State Evolution Analysis}}
We observe from
\edww{the}
LLR expressions in \eqref{eq.llr} and \eqref{eq.llr.coop} that the parameter $\tau_{t}$ plays an important role, whose converged value $\tau_{\infty}$ is determined by the fixed point of the state evolution in \eqref{eq.se} \edzr{for massive MIMO, or \eqref{eq.se.2} for cooperative MIMO.} This subsection aims to derive explicit expressions for the state evolution to evaluate $\tau_{\infty}$ for both cases.

\subsubsection{Massive MIMO}
The variance of the noise-plus-interference
\edww{term in \eqref{eq.se}}
, $\tilde{\sigma}_w^2$, can be obtained by
first computing the second-order statistics of the inter-cell interference as
\begin{align}\label{eq.inf}
\mathbf{C}_{\mathrm{inf}}&\triangleq\mathbb{E}\left[\mathrm{vec}\left(\sum_{j\neq b}^{}\mathbf{S}_{j}\mathbf{X}_{bj}\right)\mathrm{vec}\left(\sum_{j\neq b}^{}\mathbf{S}_{j}\mathbf{X}_{bj}\right)^{*}\right]\nonumber\\
&\overset{}{=}\frac{\lambda}{L}N(B-1)\mathbb{E}\left[G_{/b}^2\right]\mathbf{I},
\end{align}
where
\edf{we use
$\mathbb{E}\left[\mathbf{s}_{jn}\mathbf{s}_{jn}^{*}\right]=L^{-1}\mathbf{I}$, $\mathbb{E}[\mathbf{x}_{bjn}^{T}(\mathbf{x}_{bjn}^{T})^{*}]=\lambda\mathbb{E}[G_{/b}^2]\mathbf{I}$,
\edww{and}
 random variable $G_{/b}$ follows the distribution of the large-scale fading coefficient $g_{bjn}, j\neq b, \forall n$ \edzr{assuming that the users outside cell $b$ are uniformly and independently located.} We then obtain the expression for $\tilde{\sigma}_w^2$ as}
\begin{align}\label{eq.noise.tin}
\tilde{\sigma}_w^2=\frac{\lambda}{L}N(B-1)\mathbb{E}\left[G_{/b}^2\right]+\sigma_{w}^2.
\end{align}

To further
evaluate $\mathbb{E}[G_{/b}^2]$, we need to characterize
\edf{the probability density function (PDF)}
of $G_{/b}$. For simplicity, we assume that cell $b$ is located at the center of the network, and approximate the region of the other $B-1$ cells as a disc-shape region around cell $b$. The distribution of the large-scale fading coefficients of the users in the disc-shape region is modeled as follows.
\begin{proposition}\label{prop1}
\edww{Consider a set of}
users
uniformly distributed in a disc-shape region of maximum radius $R_{max}$ and minimum radius $R_{min}$ with one BS located at the center.
\edww{For a user at distance $d$ from the BS, we model the
large-scale fading (in dB)}
as $\alpha+\beta\log_{10}(d)$.
\edww{Then, the}
\edf{PDF}
of the large-scale fading coefficients
\edww{of this set of users}
 is given by
\begin{align}\label{eq.dist.g}
p(g)=\frac{ag^{-\gamma}}{R_{max}^{2}-R_{min}^{2}},\quad g \in [\epsilon_{min},\epsilon_{max}]
\end{align}
where $\epsilon_{min}=10^{-(\alpha+\beta\log_{10}(R_{max}))/20}$, $\epsilon_{max}=10^{-(\alpha+\beta\log_{10}(R_{min}))/20}$, $a=40\beta^{-1}10^{-2\alpha/\beta}$, and $\gamma=40\beta^{-1}+1$.
\end{proposition}
\begin{proof}
Please see Appendix~\ref{A:propA}.
\end{proof}

By setting the values of $R_{min}$ and $R_{max}$ to the radius of one cell $R_{cell}$ and the radius of the whole network $R_{net}$, respectively, we can obtain the probability distribution function of $G_{/b}$, and further evaluate $\mathbb{E}[G_{/b}^2]$ as
\begin{align}\label{eq.inf.expr}
\mathbb{E}[G_{/b}^2] &\approx \int_{\epsilon_{1}}^{\epsilon_{2}}\frac{ag^{-\gamma}g^2}{R_{net}^{2}-R_{cell}^{2}}dg\nonumber\\
&= \frac{10^{-\alpha/10}\left(R_{cell}^{2-\beta/10}-R_{net}^{2-\beta/10}\right)}{(1-\beta/20)\left(R_{cell}^2-R_{net}^2\right)},\quad \beta>20
\end{align}
where the approximation comes from the circular coverage rather than hexagonal coverage of the cells in the network, and $\epsilon_{1}=10^{-(\alpha+\beta\log_{10}(R_{net}))/20}$, $\epsilon_{2}=10^{-(\alpha+\beta\log_{10}(R_{cell}))/20}$. Simulation results in Section~\ref{sec.simu} show that the approximation error is negligible.

We can also use \eqref{eq.dist.g} to characterize the distribution of the large-scale fading coefficients of users inside cell $b$, i.e., $G_b$, by setting the values of $R_{min}$ and $R_{max}$ to $0$ and $R_{cell}$, respectively.

\edf{Now, by using the results in \eqref{eq.noise.tin}-\eqref{eq.inf.expr}, an explicit expression for the state evolution can be obtained as stated in the following proposition.}

\begin{proposition}\label{prop2}
Consider the AMP-based user activity detection for massive MIMO system in \eqref{eq.model.tin}. Based on the statistics in \eqref{eq.distri.bg} and \eqref{eq.dist.g}, and the
state evolution \eqref{eq.se}, the matrix $\mathbf{\Sigma}_{t}$ stays as a diagonal matrix with identical diagonal entries at each iteration,
i.e., $\mathbf{\Sigma}_{t}=\tau_{t}^2\mathbf{I}$, where
\begin{align}\label{eq.se.tinexpr}
\tau_{t+1}^2=\sigma_{w}^2+\frac{N(B-1)}{\lambda^{-1} L}\mathbb{E}\left[G_{/b}^2\right]+\frac{\hat{a} N}{\lambda^{-1} L}\int_{\epsilon_{2}}^{\infty}\psi(g)dg,
\end{align}
where $\hat{a}=aR_{cell}^{-2}$, an approximation of $\mathbb{E}[G_{/b}^2]$ is given by \eqref{eq.inf.expr}, and $\psi(g)$ is defined as
\begin{align}\label{A.func.psi}
\psi(g)=\frac{g^{2-\gamma}\tau_t^2}{g^2+\tau_t^2}
+\frac{g^{4-\gamma}}{g^2+\tau_t^2}\left(1-\frac{\varphi_{M}(g^2\tau_t^{-2})}{\Gamma(M+1)}\right),
\end{align}
in which
\edf{$\varphi_{M}(s)\triangleq\int_{0}^{\infty}\frac{t^M\exp(-t)}{1+(1-\lambda)(1+s)^M\exp(-st)/\lambda}dt.$}
\end{proposition}

\begin{proof}
Please see Appendix~\ref{A:propB}.
\end{proof}

\edzr{Note that Proposition~\ref{prop2} reveals the impact of the system parameters in the multi-cell scenario on the performance of AMP} \edzs{for the massive MIMO system}. \edzr{By numerically solving the fixed-point equation \eqref{eq.se.tinexpr}, the performance can be predicted under a various of systems setups. It also can be obtained from \eqref{eq.se.tinexpr} that} as the number of cells $B$ increases, the second term in \eqref{eq.se.tinexpr} tends to a constant as
\begin{align}\label{eq.se.tinexpr.B}
\frac{N(B-1)}{\lambda^{-1} L}\mathbb{E}\left[G_{/b}^2\right]\rightarrow \left( \frac{B}{R_{net}^2}\right)\left(\frac{N}{L}\right)\frac{R_{cell}^{2-\beta/10}10^{-\alpha/10}}{\lambda^{-1} (\beta/20-1)},
\end{align}
\edzr{whose value depends on $N/L$, the density of the cells $B/R_{net}^2$, cell radius $R_{cell}$, path-loss parameters $\alpha, \beta$, and the activity probability $\lambda$.}
We can use \eqref{eq.se.tinexpr.B} together with \eqref{eq.se.tinexpr} to characterize the
\edww{limiting}
performance as the number of cells tends to infinity.

\subsubsection{Cooperative MIMO}
\edzr{We derive an explicit expression for the state evolution in \eqref{eq.se.2} for cooperative MIMO via a similar approach used for \eqref{eq.se} in Proposition~\ref{prop2}. The main
difference is that the expectation in \eqref{eq.se.2} is taken with respect to $G$, whereas the expectation in \eqref{eq.se} is taken with respect to $G_{b}$.}

\begin{proposition}\label{prop4}
Consider the AMP-based user activity detection for the cooperative MIMO system in \eqref{eq.sys2}. Based on the statistics in \eqref{eq.distri.bg} and \eqref{eq.dist.g}, and the state evolution \eqref{eq.se.2}, the matrix $\mathbf{\Sigma}_{t}$ stays as a diagonal matrix with identical diagonal entries at each iteration,
i.e., $\mathbf{\Sigma}_{t}=\tau_{t}^2\mathbf{I}$, where
\begin{align}\label{eq.se.recexpr}
\tau_{t+1}^2=\sigma_{w}^2+\frac{\check{a}NB}{\lambda^{-1} L}\int_{\epsilon_{1}}^{\infty}\psi(g)dg,
\end{align}
where $\check{a}=aR^{-2}_{net}$, $\epsilon_{1}=10^{-(\alpha+\beta\log_{10}(R_{net}))/20}$, and $\psi(g)$ is defined in \eqref{A.func.psi}.
\end{proposition}
\begin{proof}
Please see Appendix~\ref{A:CoopSE}.
\end{proof}

Through a careful comparison of \eqref{eq.se.tinexpr} and \eqref{eq.se.recexpr}, we have the following result.

\begin{proposition}\label{prop5}
Let $\tau_{\infty}^{\mathrm{TIN}}$ and $\tau_{\infty}^{\mathrm{REC}}$ denote the converged $\tau_t$ in \eqref{eq.se.tinexpr}, which corresponds to the case of treating the interference as noise, and the converged $\tau_t$ in \eqref{eq.se.recexpr}, which corresponds to the case of recovering the interference, respectively.
We have $\tau_{\infty}^{\mathrm{TIN}}> \tau_{\infty}^{\mathrm{REC}}$.
\end{proposition}
\begin{proof}
Please see Appendix~\ref{A:propD}.
\end{proof}

This result shows that recovering the inter-cell interference achieves
\edww{a}
smaller value
\edww{for}
the fixed point of the state evolution, as compared to treating
\edf{it}
as noise.
\edww{Hence, recovering the interference} can increase the reliability of the preliminary activity detection at each BS.

However, in practice BSs may not want to detect all users in the network due to the complexity involved and the amount of information needed to be acquired on the signature sequences and channel statistics of all the users. Considering these practical issues, a good strategy is that each BS detects the users from its own cell and a few neighboring cells while treating the rest of the inter-cell interference as noise. Simulation results in Section~\ref{sec.simu} show that this strategy brings non-negligible improvement over treating all inter-cell interference as noise.
\subsection{\edzr{Detection Performance Analysis}}
We are interested in two types of error probability in user activity detection, the probability of false alarm, which is defined as the probability that a user is declared to be inactive whereas it is active in reality, and the probability of missed detection, which is defined as the probability that a user is declared to be inactive whereas it is active in reality. \edzr{In this subsection, we present the error probability analysis for both massive MIMO and cooperative MIMO. For cooperative MIMO, we assume perfect LLR forwarding at each BS for analytic tractability. The performance with LLR quantization is
evaluated via simulations in Section~\ref{sec.simu.quant}.}

\subsubsection{Massive MIMO}
\edz{In the following proposition,} we compute these two error probabilities on a per-user basis since
different users may achieve different trade-offs on these two types of error depending on their large-scale fading coefficients.
\edz{
\begin{proposition}\label{prop2.1}
Consider user activity detection by AMP for the massive MIMO system with fixed number of antennas $M$ in the asymptotic regime where the number of users $N$ and the length of the signature sequence $L$ tend to infinity with fixed ratio $N/L$,
\edwn{so}
that the performance can be characterized by the state evolution \eqref{eq.se.tinexpr}. Based on the LLR test with a threshold $l_{bn}^{\mathrm{LS}}$ on the AMP matched filtered output $\|\mathbf{\tilde{x}}^{t}_{bbn}\|_2^2$, the probabilities of missed detection and false alarm of user $n$ from cell $b$ are given as, respectively
\begin{align}
P_M^{\mathrm{LS},bn} &= \Gamma^{-1}(M)\cdot\bar{\gamma}\left(M,l_{bn}^{\mathrm{LS}}(g_{bbn}^2+\tau_{t}^{2})^{-1}\right),\label{eq.pm.massive}\\
P_F^{\mathrm{LS},bn}&=1-\Gamma^{-1}(M)\cdot\bar{\gamma}\left(M,l_{bn}^{\mathrm{LS}}\tau_{t}^{-2}\right),\label{eq.pf.massive}
\end{align}
where superscript ``LS" represents large-scale antenna or massive MIMO, $\Gamma(\cdot)$ is the Gamma function, and $\bar{\gamma}(\cdot,\cdot)$ is the lower incomplete Gamma function.
\end{proposition}
\begin{proof}
Please see Appendix~\ref{A:propB.1}.
\end{proof}
}
Note that both \eqref{eq.pm.massive} and \eqref{eq.pf.massive} are in similar forms as the error probabilities in single-cell case \cite{Chen2018}, \edf{but
with substantially larger value for
$\tau_{t}$
due to the inter-cell interference.}

\subsubsection{Cooperative MIMO}
The following proposition analyzes the activity detection performance in the cooperative MIMO system
\edww{assuming}
perfect LLR forwarding, which provides a performance upper bound for practical cooperative MIMO schemes.

\edz{
\begin{proposition}\label{prop5.1}
Consider user activity detection by AMP for the cooperative MIMO system with fixed number of antennas $M$ in the asymptotic regime where the number of users $N$ and the length of the signature sequence $L$ tend to infinity with fixed ratio $N/L$,
\edwn{so}
that the performance can be characterized by the state evolution \eqref{eq.se.recexpr}. Based on the LLR test with a threshold $l_{bn}^{\mathrm{CO}}$ on the aggregated result $\sum_{j\in\mathcal{B}_{bn}}\Delta_{jbn}\|\mathbf{\tilde{x}}_{jbn}\|_2^2$ at the CU, the probabilities of missed detection and false alarm of user $n$ from cell $b$ are given as, respectively
\begin{align}
P_M^{\mathrm{CO},bn}& =\int_{\mathcal{D}^{\mathrm{CO}}}\frac{\exp\left(-\sum_{j\in\mathcal{B}_{bn}}\|\mathbf{\tilde{x}}_{jbn}\|_2^2/(g_{jbn}^2+\tau_{t}^2)\right)}{\prod_{j\in\mathcal{B}_{bn}}\pi^{M}(\tau_{t}^2+g_{jbn}^2)^M}d\mathbf{\tilde{x}}_{bn},
\label{eq.pm_coop}\\
P_F^{\mathrm{CO},bn}&
=\int_{/\mathcal{D}^{\mathrm{CO}}}\frac{\exp\left(-\sum_{j\in\mathcal{B}_{bn}}\|\mathbf{\tilde{x}}_{jbn}\|_2^2/\tau_{t}^{2}\right)}{(\pi^{M}\tau_{t}^{2M})^{B_{bn}}}d\mathbf{\tilde{x}}_{bn},\label{eq.pf_coop}
\end{align}
where superscript ``CO" represents cooperative MIMO, $\mathcal{D}^{\mathrm{CO}}\triangleq \left\{\sum_{j\in\mathcal{B}_{bn}}\Delta_{jbn}\|\mathbf{\tilde{x}}_{jbn}\|_2^2 < l_{bn}^{\mathrm{CO}}\right\}$ is the decision region, and $/\mathcal{D}^{\mathrm{CO}}$ is the complementary region of $\mathcal{D}^{\mathrm{CO}}$.
\end{proposition}
\begin{proof}
Please see Appendix~\ref{A:propD.1}.
\end{proof}
}

Note that although the integral in \eqref{eq.pm_coop} is in a complicated form, it can be calculated recursively as follows. Suppose that the indices of the BSs in $\mathcal{B}_{bn}$ are $1,\cdots,B_{bn}$. By transforming the Cartesian coordinates into the spherical coordinates, the integral can be re-written as
\begin{align}\label{eq.pm_coop_exp}
&P_M^{\mathrm{CO},bn}\nonumber\\
&=\underbrace{
\idotsint\limits_{\sum \theta_jr_j<l}\exp\left(-\sum_{j=1}^{B_{bn}}r_j\right)
r_1^{M-1}dr_{1}\cdots r_{B_{bn}}^{M-1}dr_{B_{bn}}}
_{f_M(l;\theta_1,\cdots,\theta_{B_{bn}})}\nonumber\\
&=\int_{0}^{l/\theta_1}\exp(-r_1)r_1^{M-1}dr_1\cdot\nonumber\\
&\underbrace{
\idotsint\limits_{\sum_{j\neq 1} \theta_jr_j<l-\theta_1r_1}\exp\left(-\sum_{j=2}^{B_{bn}}r_j\right)
r_2^{M-1}dr_{2}\cdots r_{B_{bn}}^{M-1}dr_{B_{bn}}}
_{f_{M-1}(l-\theta_1r_1;\theta_2,\cdots,\theta_{B_{bn}})},
\end{align}
where we define $\theta_j\triangleq \theta_{jbn}$ (recall $\theta_{jbn}=g_{jbn}^2\tau_{t}^{-2}$) and $l\triangleq l_{bn}^{\mathrm{CO}}$ for \edz{notational} simplicity. We observe that \edz{$P_M^{\mathrm{CO},bn}$} can be obtained
by calculating a sequence of functions $f_{M}(\cdot), f_{M-1}(\cdot),\cdots ,f_{1}(\cdot)$ recursively, and each of them has a
\edww{closed-form}
expression due to the fact that each integrand is constituted
\edww{of}
polynomial functions and exponential functions. Similarly, the probability of false alarm in \eqref{eq.pf_coop} can be evaluated by using the same approach as in \eqref{eq.pm_coop_exp}.

The error probability expressions in \eqref{eq.pm.massive}-\eqref{eq.pf.massive} and \eqref{eq.pm_coop}-\eqref{eq.pf_coop} can be used to characterize the detection performance of massive MIMO and cooperative MIMO, respectively. However, since the expressions are complicated, it is difficult to perform an analytic comparison between the performance of these two architectures. Instead, we numerically evaluate the error probabilities and demonstrate the comparison under different parameter settings in Section~\ref{sec.simu}.

\subsection{\edzr{Asymptotic Performance Analysis}}
\subsubsection{Massive MIMO}
We
\edww{now}
study the asymptotic performance of the user activity detection in the massive MIMO system as $M$ tends to infinity, which is an extension of the previous result in \cite{Liu2018}, where the asymptotic analysis for single-cell systems is carried out.

\begin{proposition}\label{prop3}
\edz{Consider user activity detection by AMP for the massive MIMO system with fixed number of antennas $M$ where the probabilities of missed detection and false alarm of user $n$ from cell $b$ are given in \eqref{eq.pm.massive} and \eqref{eq.pf.massive}, respectively.} If we further let $M$ tend to infinity, we can achieve perfect detection, i.e.,
\begin{align}
\lim_{M\rightarrow \infty}P_M^{\mathrm{LS},bn} = \lim_{M\rightarrow \infty}P_F^{\mathrm{LS},bn}=0, \forall b,n,
\end{align}
by properly
\edww{choosing}
the threshold $l_{bn}^{\mathrm{LS}}$ in the LLR test.
\end{proposition}
A rigours proof can be obtained by following the idea in \cite{Liu2018}.
\edff{The following is}
an alternative way to show the asymptotic behavior by
using the central limit theorem (CLT).

\edz{Note that as $M$ increases, the matched filtered output $\|\mathbf{\tilde{x}}^{t}_{bbn}\|_2^2$ tends to a mixed Gaussian distribution by CLT depending on whether the user is active or not. The probability of missed detection then can be approximated by the cumulative distribution function (CDF) of a Gaussian distribution as}
\begin{align}\label{eq.pm.massive.simp}
P_M^{\mathrm{LS},bn}\approx \Phi\left(\frac{l_{bn}^{\mathrm{LS}}/(g_{bbn}^2+\tau_{t}^2)-M}{\sqrt{M}}\right),
\end{align}
where $\Phi(u)\triangleq (2\pi)^{-1/2}\int_{-\infty}^{u}\exp(-t^2/2)dt$ is the CDF of the standard Gaussian distribution. Similarly, the probability of false alarm can be approximated as
\begin{align}\label{eq.pf.massive.simp}
P_F^{\mathrm{LS},bn}\approx 1- \Phi\left(\frac{l_{bn}^{\mathrm{LS}}/\tau_{t}^2-M}{\sqrt{M}}\right).
\end{align}
\edf{Now,} if we set the threshold $l_{bn}^{\mathrm{LS}}$ to be $\mu M$, where $\mu$ is a constant selected from the interval $(\tau_{t}^2, \tau_{t}^2+g_{bbn}^2)$, both \eqref{eq.pm.massive.simp} and \eqref{eq.pf.massive.simp} tend to zero as $M$ tends to infinity.

\subsubsection{Cooperative MIMO}
We briefly discuss the asymptotic behavior of the user activity detection for the cooperative MIMO system. We are interested in the limit where the number of cooperative BSs $B_{bn}$ increases to infinity with the number of antennas per BS fixed.

\edww{Interestingly,}
perfect detection cannot be achieved in the regime where $B_{bn}$ tends to infinity. As an intuitive explanation, we re-write \eqref{eq.pm_coop} by introducing a sequence of i.i.d.\ random variables $X_j, j\in \mathcal{B}_{bn}$ all following $\chi^2$ distribution with $2M$ degree of freedom as
\begin{align}
P_M^{\mathrm{CO},bn}=\mathrm{Pr}\left(\sum_{j\in\mathcal{B}_{bn}}\frac{g_{jbn}^2}{2\tau_t^2}X_j < l_{bn}^{\mathrm{CO}}\right),
\end{align}
where the coefficient $g_{jbn}^2/(2\tau_t^2)$ can be seen as a measure of the contribution from $X_j$. Note that for user $n$ from cell $b$, the sum $\sum_{j\in\mathcal{B}_{bn}} g_{jbn}^2/(2\tau_t^2)$ converges to a finite value as $B_{bn}$ tends to infinity due to the attenuation of $g_{jbn}$, indicating that $P_M^{\mathrm{CO},bn}$ is dominated by only a few terms in $\sum_{j\in\mathcal{B}_{bn}} g_{jbn}^2/(2\tau_t^2)$ with large coefficients. Similar analysis applies to the probability of false alarm. Although perfect detection cannot be achieved in the cooperative MIMO system as $B_{bn}$ tends to infinity, cooperative MIMO is expected to offer performance improvement to
cell-edge users as $B_{bn}$ increases. We illustrate this point by simulations in Section~\ref{sec.simu}.

\section{Simulation Results}
\label{sec.simu}
We consider a network comprising $19$ hexagonal cells placed in three tiers.
There are in total $2000\times 19=38000$ potential users uniformly distributed in the network, among which five percent are active. The BS-to-BS distance is 2000m. The path-loss from BSs to users is modeled as $15.3+37.6\log_{10}(d)$, where $d$ is BS-user distance measured in meter. The transmit power of each user is $23$dBm, and the background noise power is $-169$dBm/Hz over $10$MHz. The length of signature sequences is set as $400$ unless otherwise specified. \edzr{Regarding the signature length, it is worth mentioning that in practice the maximum possible length is limited by the dimensions of the channel coherence block. Typical values of the dimensions could be a few hundreds to a few thousands, depending on the wireless propagation environment \cite{Bjoernson2016}. By noting that the potential users in mMTC are usually low-mobility and low-rate devices, it is very likely that the coherence block could be large enough to support long signature sequences transmission for massive random access.}

We illustrate the performance of the users in the innermost cell
that endure the severest inter-cell interference. \edzr{By adjusting the threshold in the LLR detection for each user in the cell, a trade-off between the probability of false alarm and the probability of missed detection can be obtained. An example of the trade-off curves of two different users in a non-cooperative MIMO system is illustrated in Fig.~\ref{fig.tradeoff}, where the curves with legend ``predicted" are obtained from the analytical characterization, whereas those with legend ``simulated" are obtained by running the AMP algorithm. It can be seen that the simulation agrees with the analysis very well. In practice, as the cost of missed detection may be different from the cost of false alarm, the threshold for each user needs to be determined by an analysis of such cost. For the convenience of demonstration, in the simulations in the rest of this section we select the threshold in such a way that the probability of false alarm and the probability of missed detection are equal and label the probability as ``probability of false alarm/missed detection". Note that users at different locations have different probabilities of error. For this reason, we plot the CDF of the probabilities of false alarm/missed detection. In the CDF curve, we are particularly interested in the probability of false alarm/missed detection of the cell-edge user, which is defined as the 95-percentile point of the CDF.}

\begin{figure}
\centerline{\epsfig{figure=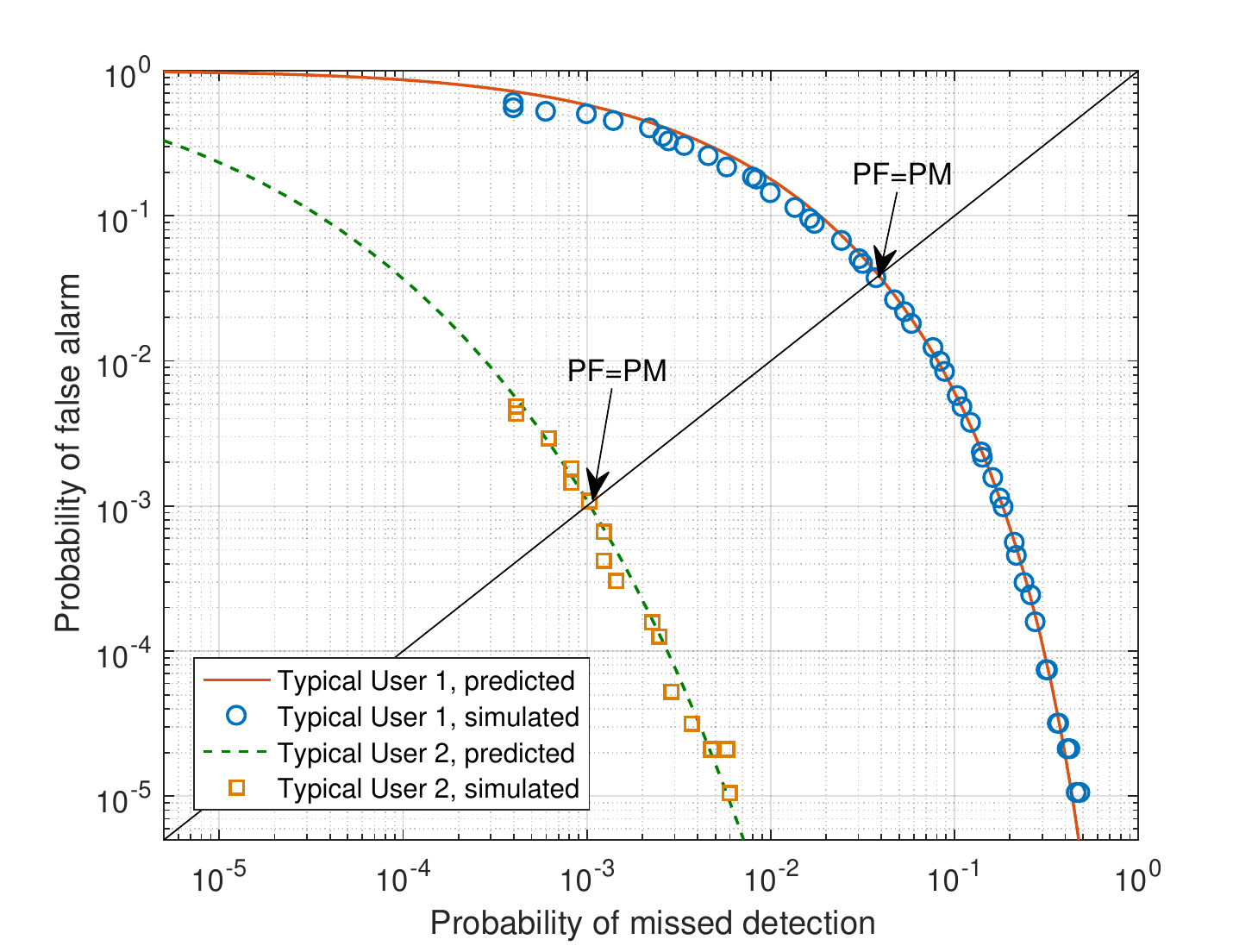,width=0.5\textwidth}}
\caption{\edzs{Probability of false alarm versus probability of missed detection in a non-cooperative massive MIMO system with $M=8$. Typical User 1 and Typical User 2
correspond to the user at 95-percentile and 50-percentile of the CDF curve, respectively.}}
\label{fig.tradeoff}
\end{figure}

\subsection{Validation of Analytic Results}
\label{sec.simu.vali}
We first validate our performance analysis for massive MIMO and cooperative MIMO via Monte Carlo simulations. The results with legend ``predicted" are obtained from the analytical characterization of the detection error in Section~\ref{sec.massive}, whereas those with legend ``simulated" are obtained by running the AMP algorithm over $10^5$ channel realizations.

\subsubsection{Massive MIMO}
\edzr{Fig.~\ref{fig.massive} shows the CDF of the probability of false alarm/missed detection in the massive MIMO scenario, from which we observe that the simulation results match the analytical results very well under different numbers of antennas.}
Note that the predicted results are obtained with the circular coverage approximation used in \eqref{eq.inf.expr}, indicating that the effect of the approximation is indeed negligible.
\edzr{We would also like to remark that, unlike the conventional massive MIMO where the number of users is assumed much smaller than the number of antennas, in the} \edzs{mMTC scenario} \edzr{the number of active users could be comparable or even larger than the number of antennas due to the large user pool. However, it can be seen from Fig.~\ref{fig.massive} that as $M$ increases, the detection errors of all users drop rapidly: given about $100$ active users per cell and $32$ antennas per BS, more than $90$ percent of the users achieve error probability less than $10^{-4}$, and almost all users achieve error probability less than $10^{-3}$.}

\begin{figure}
\centerline{\epsfig{figure=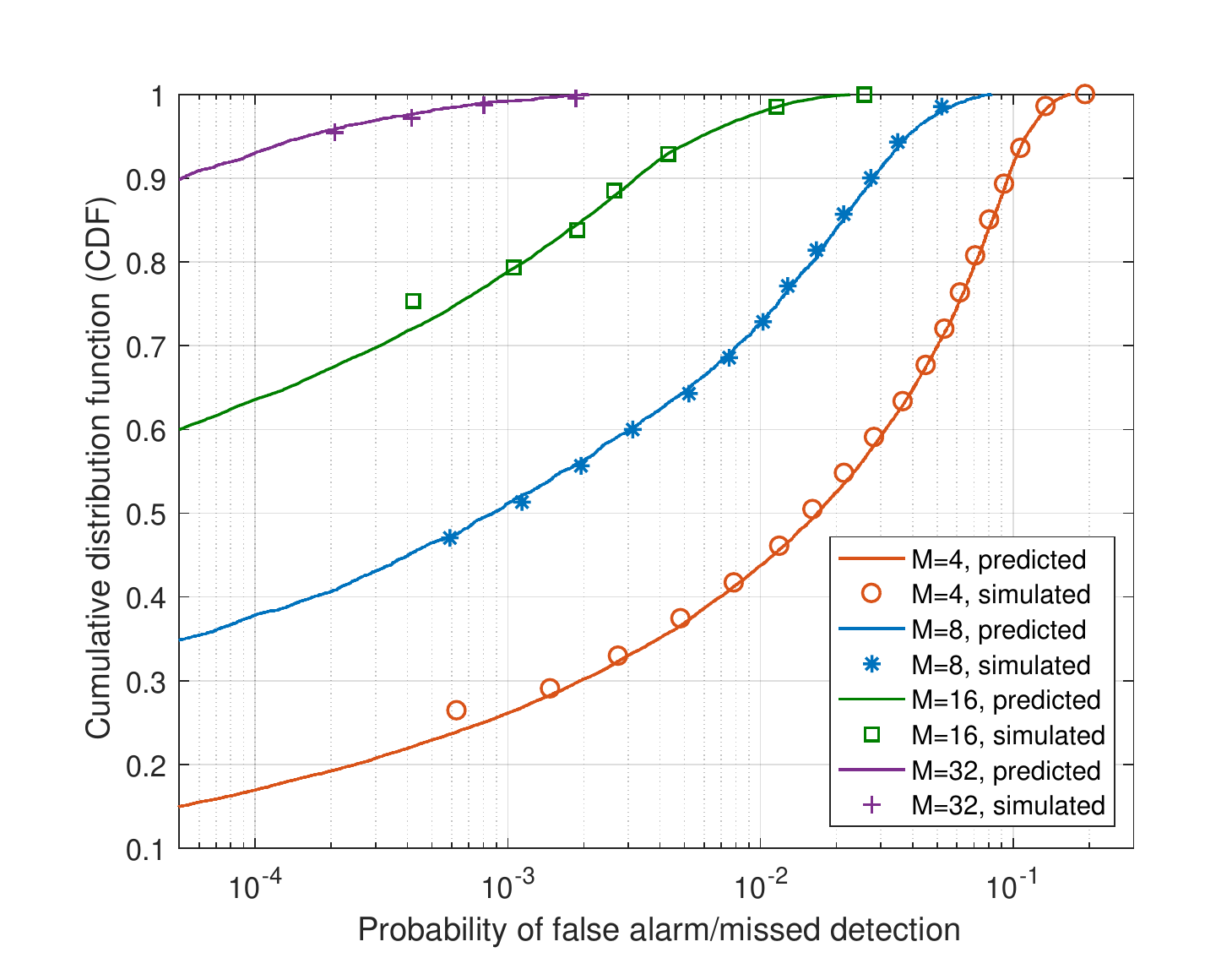,width=0.5\textwidth}}
\caption{\edzr{CDF of the detection error of the users in the innermost cell for the massive MIMO architecture.}}
\label{fig.massive}
\end{figure}

\subsubsection{Cooperative MIMO}
Before validating the performance analysis for the cooperative MIMO case, we first aim to illustrate
that recovering the inter-cell interference instead of treating it as noise indeed brings improvement, as stated in Proposition~\ref{prop5}. We consider the innermost BS and run AMP for several different detection ranges. We gradually increase the detection range of the BS to recover more inter-cell interference, while treating the signal from beyond the range as noise. Fig.~\ref{fig.tinrec} depicts how the value of $\tau_{\infty}^{2}$ in AMP changes with different choices of the range. Fig.~\ref{fig.tinrec} shows that enlarging the detection range to recover more inter-cell interference helps reduce $\tau_{\infty}^{2}$ in AMP, which improves the reliability of the preliminary activity detection at each BS. However, the improvement becomes negligible when the detection range exceeds a certain point whose value depends on $M$ and $L$. This is because the interference caused by out-of-cell users that are sufficiently far away from the BS is weak so that it can be treated as noise without any significant performance degradation. We observe that increasing $M$ or $L$ helps raise the point, indicating that the BS becomes more capable of recovering the inter-cell interference.

\begin{figure}
\centerline{\epsfig{figure=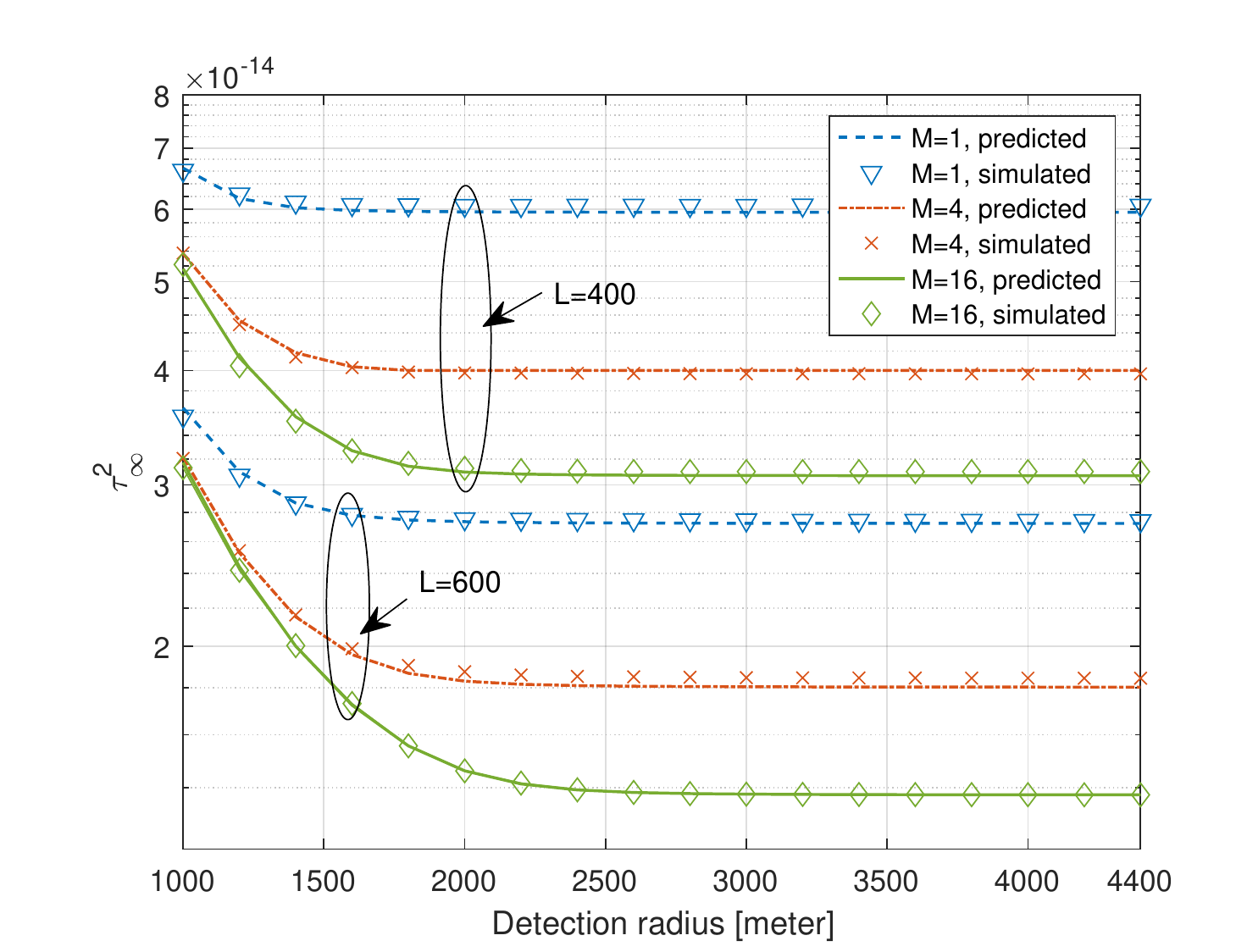,width=0.5\textwidth}}
\caption{The values of $\tau_{\infty}^{2}$ under different detection ranges of the innermost BS for cooperative MIMO.}
\label{fig.tinrec}
\end{figure}

Fig.~\ref{fig.coop} shows the CDF of the detection error for the cooperative MIMO system, where each BS detects the users from its own cell as well as six surrounding cells, and each BS is able to forward a subset of LLRs of the detected users to the CU perfectly. The legend ``Coop. w/ 3 BSs" indicates that the number of BSs involved in the LLR aggregation for each user is three, i.e., $B_{bn}=3, \forall b,n$. We plot the case ``Coop. w/ 1 BS" to provide a performance baseline to illustrate the impact of $B_{bn}$. Note that although only one BS is involved in the LLR aggregation, the case ``Coop. w/ 1 BS" still differs from the massive MIMO scenario by recovering the inter-cell interference rather than treating it as noise. We observe from Fig.~\ref{fig.coop} that the simulation results match the predicted results for different $M$ and $B_{bn}$. Compared to ``Coop. w/ 1 BS", by adding one more BS in the LLR aggregation, ``Coop. w/ 2 BSs" substantially improves the cell-edge user performance, showing the necessity of LLR aggregation (of more than one BS). However, the improvement by further increasing $B_{bn}$ is less prominent, indicating that it is enough to set $B_{bn}$ as a small number, such as two or three, to exploit most benefit of cooperation. It
\edww{can}
be observed that as $M$ increases, the performance gaps become more prominent.

\begin{figure}
\centerline{\epsfig{figure=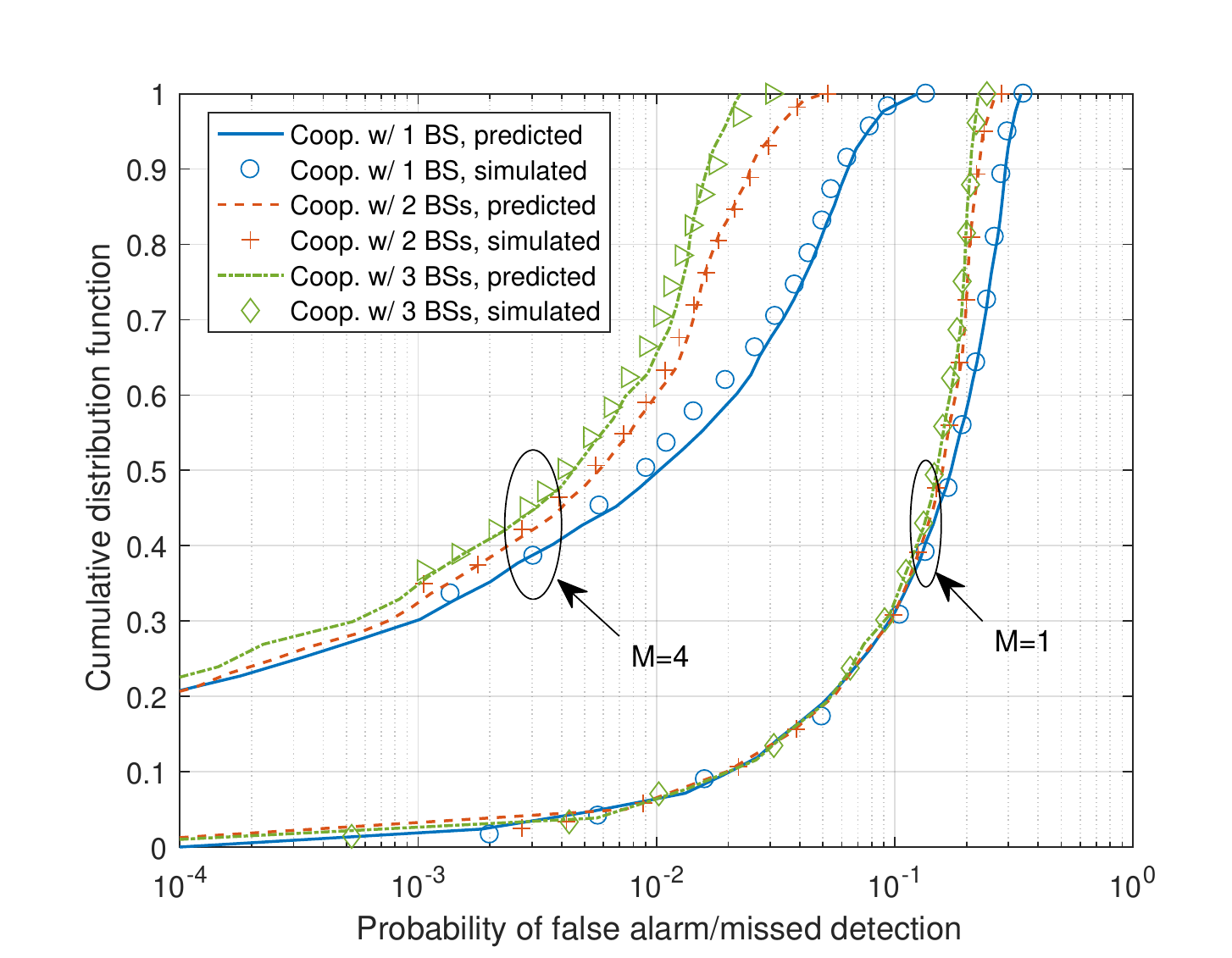,width=0.5\textwidth}}
\caption{CDF of the detection error of the users in the innermost cell for the cooperative MIMO architecture.}
\label{fig.coop}
\end{figure}

\subsection{Massive MIMO vs Cooperative MIMO for User Activity Detection}
\label{sec.simu.compare}
This section exploits our analytic results to numerically compare the detection performance of massive MIMO and cooperative MIMO. We are mostly interested in the cell-edge user performance, which is defined here as the 95-percentile point of the CDF. Fig.~\ref{fig.massivecoopcdf} plots the cell-edge user performance in massive MIMO and cooperative MIMO, where the number of the cooperative BSs, $B_{bn}$, increases from one to four for cooperative MIMO. The number of antennas per BS for cooperative MIMO is set as $4,8,12$, or $16$, while a larger array with $16,32,48$, or $64$ antennas is considered for massive MIMO. Fig.~\ref{fig.massivecoopcdf} clearly shows that as the number of cooperating BSs increases from one to three, cooperative MIMO brings substantial reduction in detection error. The performance improvement then becomes marginal with four cooperative BSs, indicating that cooperating three BSs is already enough to achieve most of the benefits of cooperative MIMO.

\rr{We also observe from Fig.~\ref{fig.massivecoopcdf} that by increasing the number of cooperating BSs, cooperative MIMO with $4,8,12$ or $16$ antennas approximately approaches massive MIMO with $16,32,48$, or $64$ antennas, respectively. Although massive MIMO still achieves sightly lower error probabilities, cooperative MIMO has its advantage in substantially reducing the number of antennas needed at each BS. To compare the effectiveness of exploiting cooperation versus large number of antennas more clearly, we plot the cell-edge user performance of cooperative MIMO and massive MIMO with an increasing number of antennas at the BS in Fig.~\ref{fig.massivecooptp}. The number of BSs in cooperation for cooperative MIMO is set as three, i.e., $B_{bn}=3$, Fig.~\ref{fig.massivecooptp} first confirms that as $M$ increases, the detection error drops almost exponentially in massive MIMO, indicating that by employing a sufficiently large number of antennas, the detection error can be driven to zero. Fig.~\ref{fig.massivecooptp} also shows that in terms of cell-edge user performance, given a fixed total number of antennas involved in the detection, massive MIMO is not as efficient as cooperative MIMO in exploiting multiple antennas. For example, cooperating three BSs with $16$ antennas at each BS achieves better performance than massive MIMO with $48$ antennas at each BS. The improvement mainly comes from the benefit of recovering the inter-cell interference rather than treating the interference as noise in cooperative MIMO.}

\begin{figure}
\centerline{\epsfig{figure=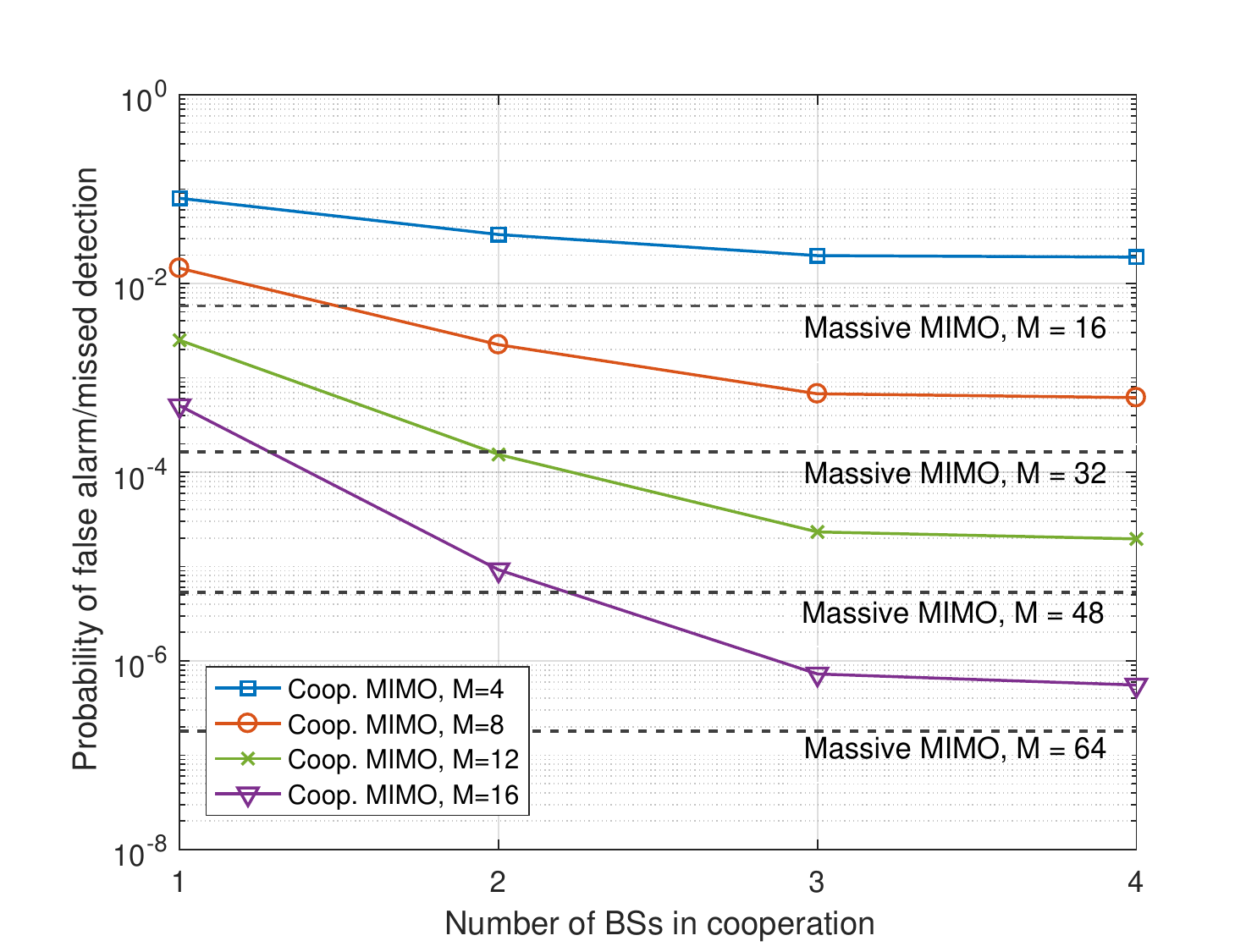,width=0.5\textwidth}}
\caption{
\edzr{Comparison of massive MIMO and cooperative MIMO with an increasing number of cooperative BSs in terms of cell-edge user performance.}}
\label{fig.massivecoopcdf}
\end{figure}

\begin{figure}
\centerline{\epsfig{figure=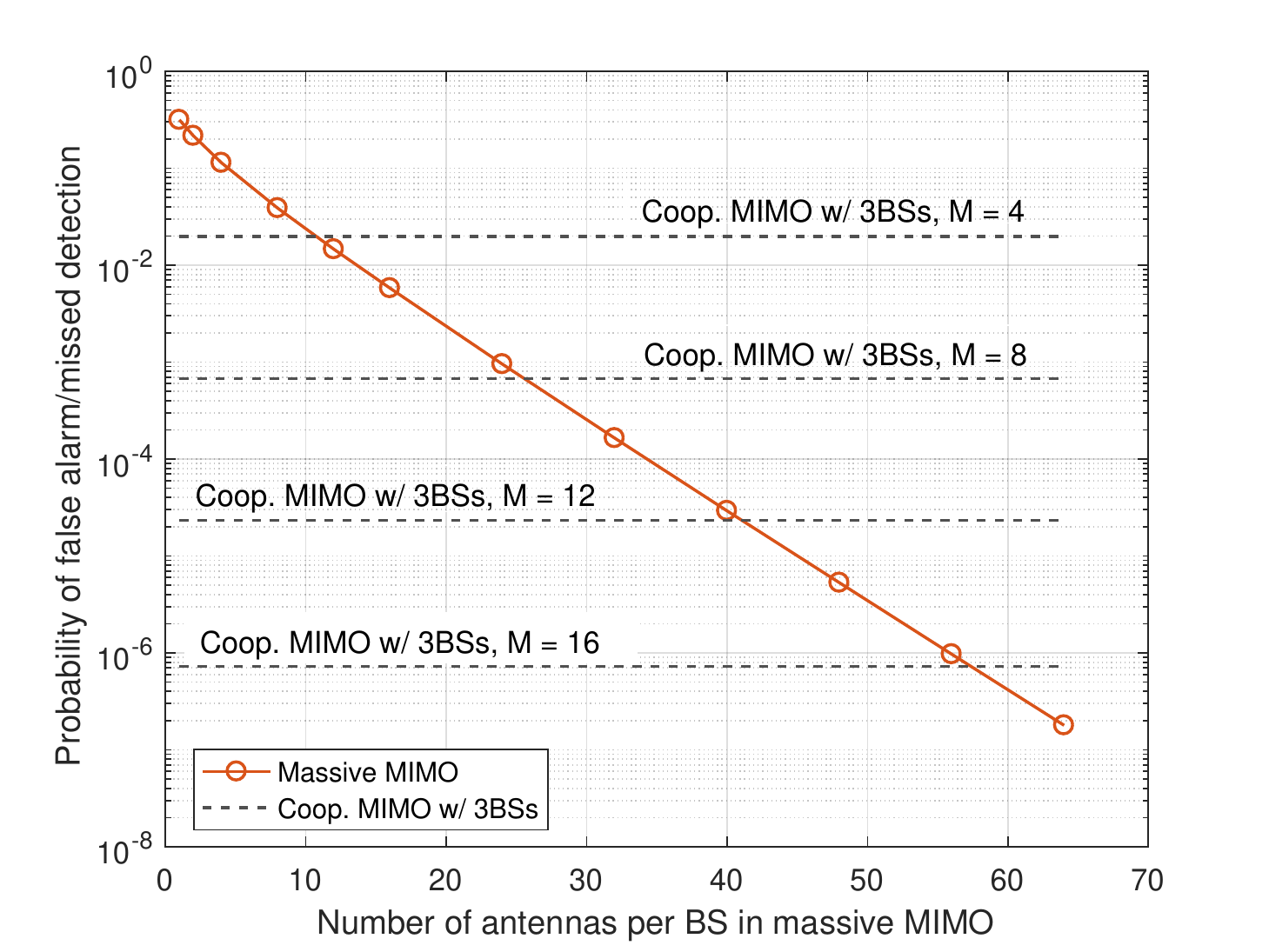,width=0.5\textwidth}}
\caption{
\edzr{Comparison of cooperative MIMO and massive MIMO with an increasing number of BS antennas in terms of cell-edge user performance.}}
\label{fig.massivecooptp}
\end{figure}

\edzr{Fig.~\ref{fig.massivecoopl} shows the cell-edge user performance versus the length of the signature sequence, $L$, in massive MIMO and cooperative MIMO with three-BS cooperation, i.e., $B_{bn}=3$.} We observe that increasing $L$ efficiently improves the detection performance, which is due to the reduction of $\tau_{\infty}^2$ in AMP. \edzr{The improvement is more substantial in cooperative MIMO because of the recovery of the inter-cell interference.
Further, we observe that} \edzs{for this specific case, a three-BS} \edzr{cooperative MIMO system achieves a comparable detection performance to a massive MIMO system by using only around one fourth of the antennas per BS, when $L$ is large.}

\begin{figure}
\centerline{\epsfig{figure=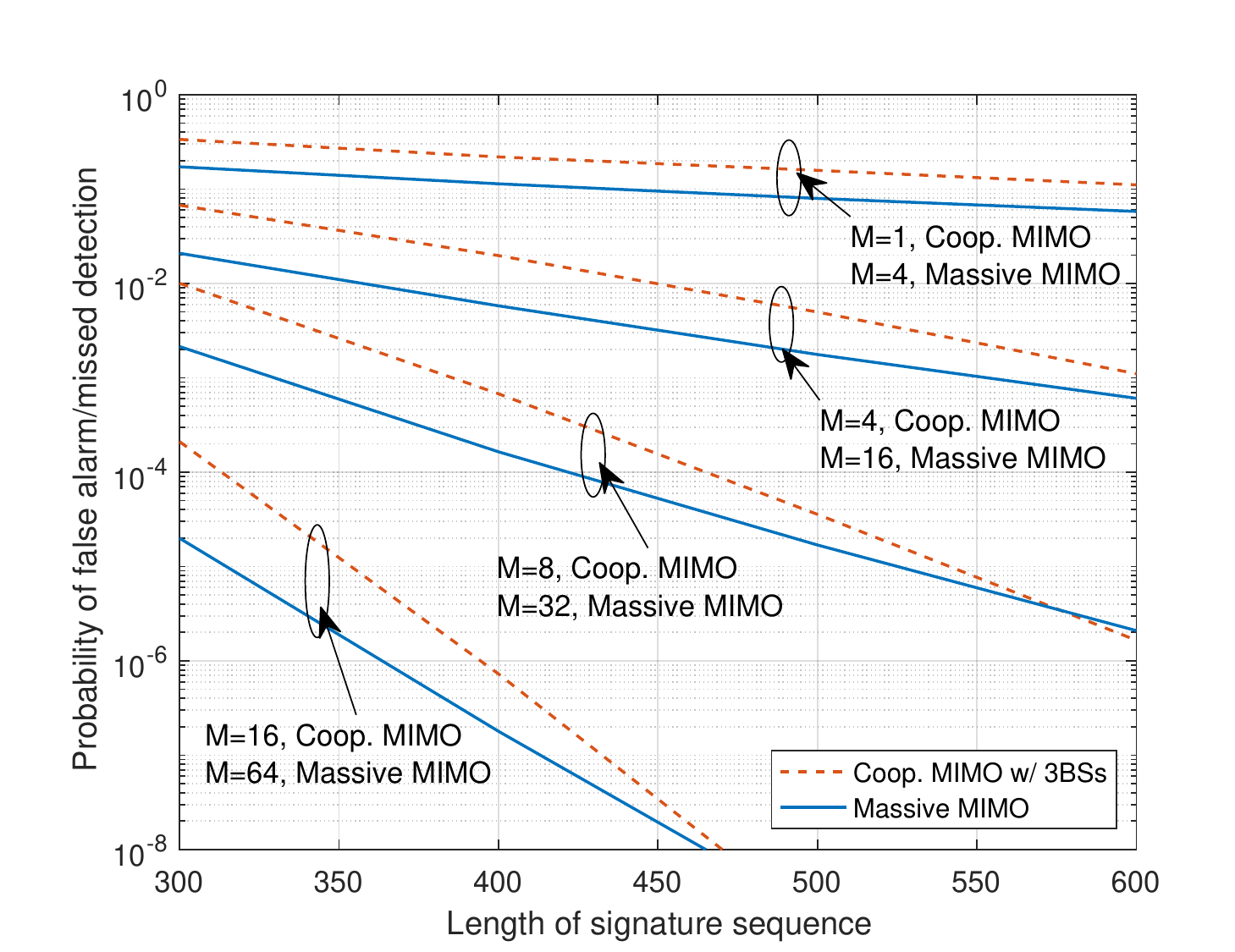,width=0.5\textwidth}}
\caption{\edzr{Cell-edge user performance versus the length of the signature sequence.}}
\label{fig.massivecoopl}
\end{figure}

\subsection{\edww{Impact of Limited Fronthaul}}
\label{sec.simu.quant}
\edzr{We have so far compared massive MIMO and cooperative MIMO with ideal LLR forwarding. In this section, we investigate the impact of limited
fronthaul capacity
on cooperative MIMO}.
\edzr{We show that the performance of ideal LLR forwarding can be efficiently approached by the} \edzs{quantization scheme proposed in Section~\ref{sec.quant}}. Here, we set the cooperation size as $B_{bn}=3, \forall b,n,$
\edf{so that if each LLR value is quantized to $Q$ bits,
the amount of information bits needed to be forwarded via the fronthaul link from each BS is about $3NQ$ \edff{bits}
\edwn{per signature sequence length.}}
Fig.~\ref{fig.q1} plots the CDF of the detection error when $M=1$ for different choices of quantization bits $Q$ and the percentage of the quantization coverage $\zeta$. Fig.~\ref{fig.q1} shows that for  single-antenna BSs, the parameter $\zeta=0.95$ is preferred to $\zeta=0.99$. This is because the generated quantization levels with $\zeta=0.99$ fail to capture the small LLR values. Further, we observe that with $\zeta=0.95$, the cooperative MIMO system with only $3$ quantization bits per LLR can already approach the performance of the ideal LLR forwarding case. This illustrates the effectiveness of the proposed simple uniform quantization scheme.

Fig.~\ref{fig.q2} plots the detection performance with different $Q$ and $\zeta$ when each BS is equipped with $M=4$ antennas. Compared to Fig.~\ref{fig.q1}, we observe that a slightly larger $\zeta=0.97$ is preferred, which is due to the shrinkage of the tail probability of the LLR distribution for larger $M$. It can be seen from Fig.~\ref{fig.q2} that with $\zeta = 0.97$, the proposed scheme with $Q=4$ can approach the performance of the ideal LLR forwarding case.

\begin{figure}
\centerline{\epsfig{figure=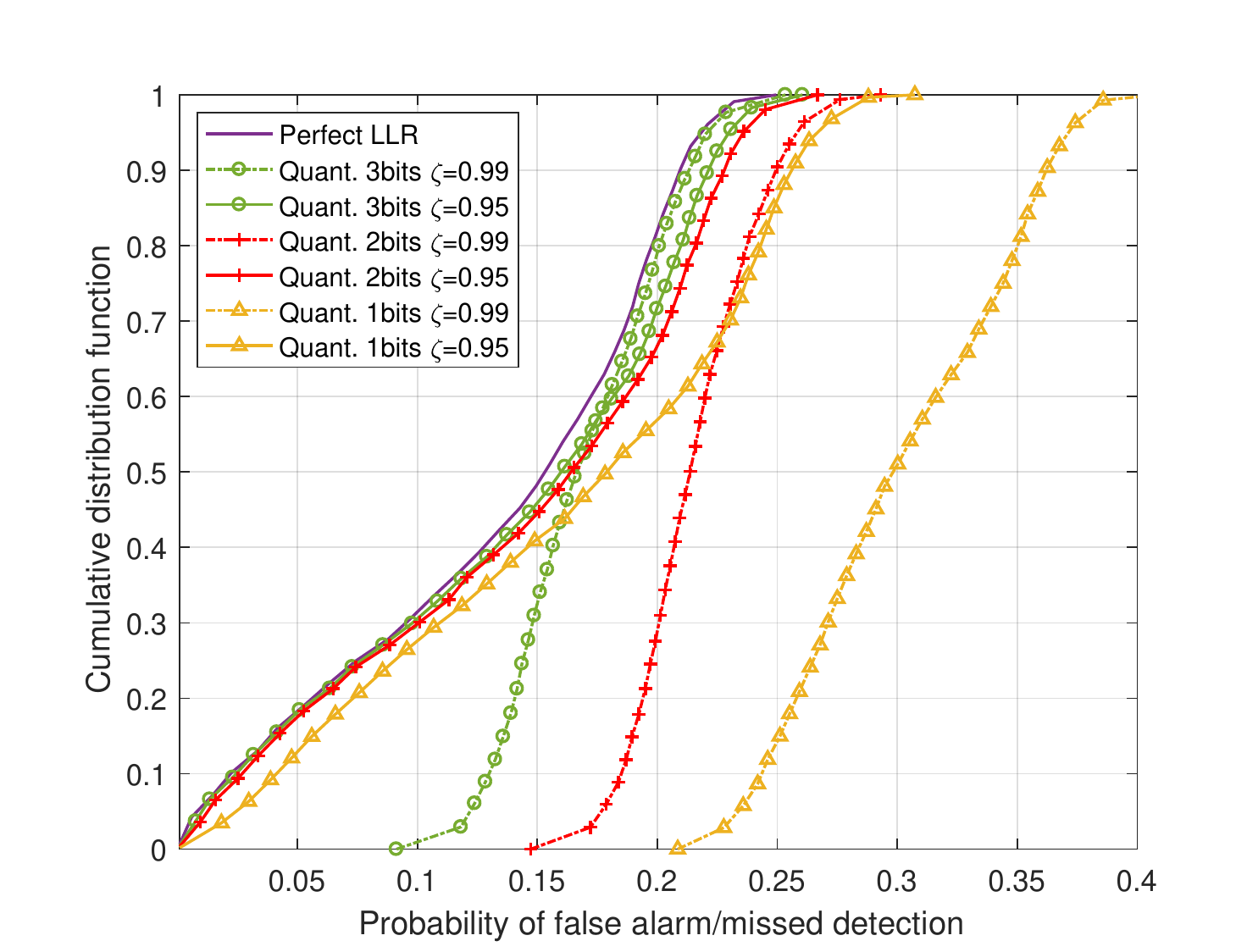,width=0.5\textwidth}}
\caption{CDF of the detection error with LLR quantization in cooperative architecture when $M=1$.}
\label{fig.q1}
\end{figure}

\begin{figure}
\centerline{\epsfig{figure=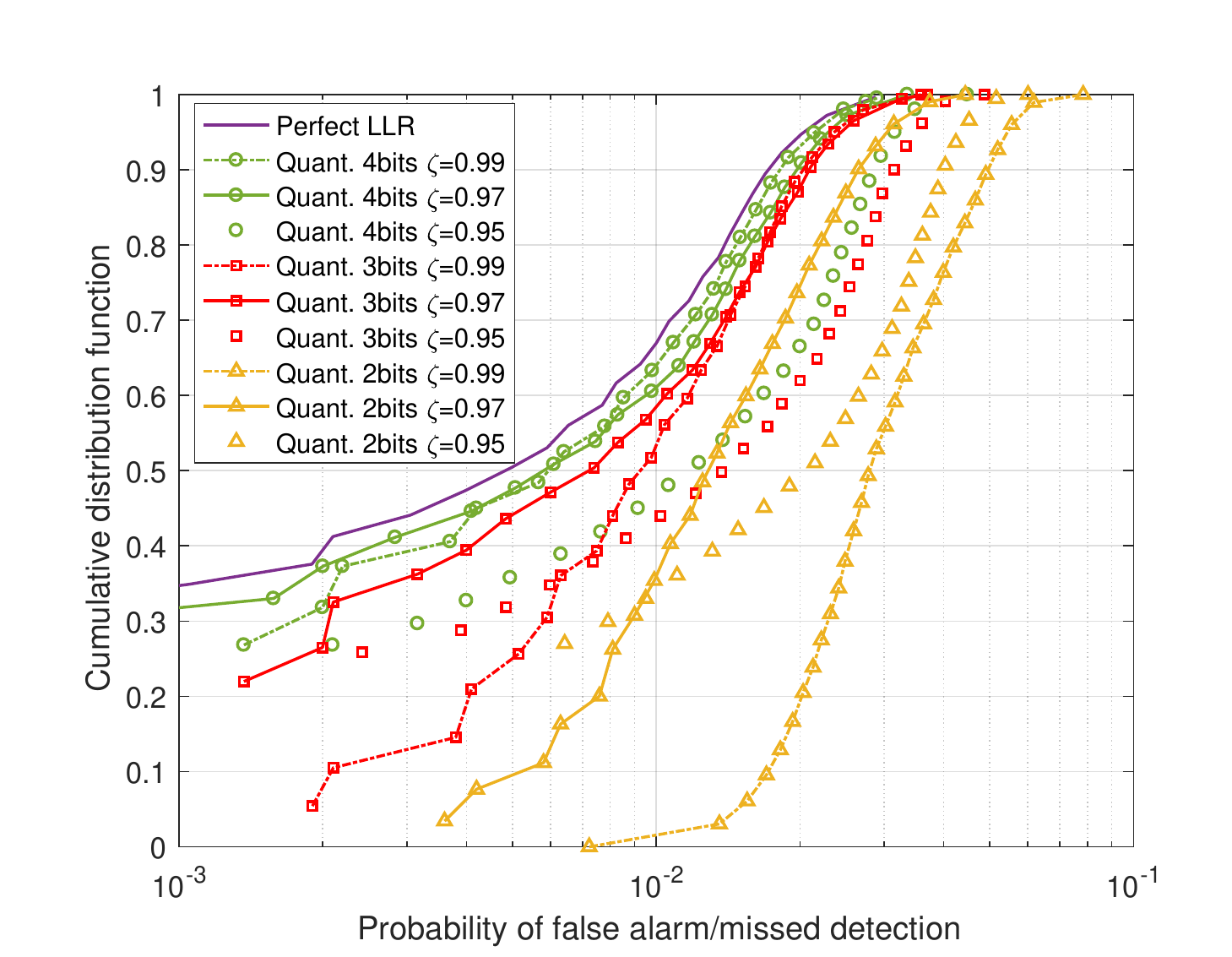,width=0.5\textwidth}}
\caption{CDF of the detection error with LLR quantization in cooperative architecture when $M=4$.}
\label{fig.q2}
\end{figure}

\section{Conclusion}
\label{sec.con}
This paper studies the multi-cell user activity detection problem for massive connectivity, in which inter-cell interference is a limiting factor for accurate detection. Two promising network architectures, massive MIMO and cooperative MIMO, are compared in terms of their effectiveness in combating inter-cell interference, while employing the computationally efficient AMP algorithm. This paper characterizes the detection performances for both network architectures. Numerical results validate the analytic analysis, and demonstrate that massive MIMO effectively improves the performance of all users as the number of antennas per BS increases, whereas cooperative MIMO mainly improves the performance of cell-edge users as the size of the cooperation increases. \edzs{Numerical simulations of a specific practical scenario suggest that in terms of the cell-edge user performance, a cooperative MIMO system with three cooperating BSs is as effective as a non-cooperative massive MIMO system with four times as many antennas in overcoming inter-cell interference.}

\appendix
\subsection{Proof of Proposition \ref{prop1}}
\label{A:propA}
Let $d$ denote the user-BS distance. Since users are uniformly distributed,
\edf{the PDF}
of $d$ is
\begin{align}\label{A.eq.d}
p(d)=\frac{2d}{R_{max}^2-R_{min}^2},\quad R_{min}\leq d\leq R_{max}.
\end{align}
Using \eqref{A.eq.d},
\edf{the PDF}
of the large-scale fading coefficient $g=10^{-(\alpha+\beta\log_{10}d)/20}$ is given as
\begin{align}
p(g)&=\left(\frac{40}{\beta}\right)\frac{10^{-2\alpha/\beta}g^{-40/\beta-1}}{R_{max}^2-R_{min}^2}, \,\,\,\epsilon_{min} \leq g \leq \epsilon_{max},
\end{align}
where $\epsilon_{min}=10^{-(\alpha+\beta\log_{10}R_{max})/20}$ and $\epsilon_{max}= 10^{-(\alpha+\beta\log_{10}R_{min})/20}$.
\edf{By defining $a\triangleq 40\beta^{-1}10^{-2\alpha/\beta}$, $\gamma \triangleq 40/\beta+1$, we then obtain \eqref{eq.dist.g}. Here, we omit the \edww{shadowing component}
for simplicity. }

\subsection{Proof of Proposition \ref{prop2}}
\label{A:propB}
We only provide
\edww{the}
key steps here. A similar
derivation can be found in \cite{Chen2018}.
By induction, we assume that $\mathbf{\Sigma}_{t}$ is a diagonal matrix with identical diagonal entries, i.e., $\mathbf{\Sigma}_{t}=\tau_t^2\mathbf{I}$. We compute the covariance matrix $\mathbf{C}^{t}=\mathbb{E}[\mathbf{D}^{t}(\mathbf{D}^{t})^{*}|\mathbf{\tilde{R}}^{t},G_b]$ given $\mathbf{\tilde{R}}_t=\mathbf{\tilde{r}}_t$ and $G_b=g$ as
\begin{align}
\mathbf{C}^{t}(\mathbf{\tilde{r}}_t,g)=\frac{g^2\tau_t^2\phi^{-1}(\mathbf{\tilde{r}}_t)}{g^2+\tau_t^2}\mathbf{I}+\frac{\phi^{-1}(\mathbf{\tilde{r}}_t)-\phi^{-2}(\mathbf{\tilde{r}}_t)}{g^{-4}(g^2+\tau_t^2)^2}\mathbf{\tilde{r}}_{t}^{T}(\mathbf{\tilde{r}}_{t}^{T})^{*},
\end{align}
where $\phi(\mathbf{\tilde{r}}_t)\triangleq 1+(1-\lambda)(1+g^2\tau_t^{-2})^{M}\exp(-\Delta\|\mathbf{\tilde{r}}_t\|_2^2)/\lambda$. We then derive $\mathbb{E}\left[\mathbf{C}^{t}\big|G_b\right]$ by taking the expectation with respect to $\mathbf{\tilde{R}}_{t}$. It can be shown that each off-diagonal entry of  $\mathbb{E}\left[\mathbf{C}^{t}\big|G_b\right]$ is zero, and each diagonal entry of $\mathbb{E}\left[\mathbf{C}^{t}\big|G_b\right]$ given $G_b=g$ can be computed as
\begin{align}\label{A.eq.C}
C=&\frac{\lambda g^2\tau_t^2}{g^2+\tau_t^2} + \frac{\lambda g^4}{g^2+\tau_t^2}\left(1-\frac{\varphi_{M}(g^2\tau_t^{-2})}{\Gamma(M+1)}\right)
\end{align}
Finally, we obtain $\mathbb{E}\left[\mathbb{E}\left[\mathbf{C}^{t}\big|G_b\right]\right]$ by taking the expectation with respect to $G_b$, which leads the third term in the right hand side of \eqref{eq.se.tinexpr}, indicating that $\mathbf{\Sigma}_{t+1}$ is a diagonal matrix with identical diagonal entries, which completes the induction.

\subsection{Proof of Proposition \ref{prop4}}
\label{A:CoopSE}
The proof to show that $\mathbf{\Sigma}_{t}$ stays as a diagonal matrix with identical diagonal entries at each iteration follows the same idea in Proposition \ref{prop2}. To derive an explicit expression for the second term in the right hand side of \eqref{eq.se.2}, we use the result in \eqref{A.eq.C}. By take the expectation of \eqref{A.eq.C} with respect to $G$ (rather than $G_b$), we get the second term in \eqref{eq.se.recexpr}. Note that
\edf{the PDF}
of $G$ can be obtained from Proposition \ref{prop1} by setting $R_{max}=R_{net}$ and $R_{min}=0$.

\subsection{Proof of Proposition \ref{prop5}}
\label{A:propD}
We directly compare the right hand side of \eqref{eq.se.tinexpr} and \eqref{eq.se.recexpr} by
breaking the integral in \eqref{eq.se.recexpr} into $\int_{\epsilon_{1}}^{\epsilon_{2}}\psi(g)dg$ and $\int_{\epsilon_{2}}^{\infty}\psi(g)dg$, and noticing $\check{a}B=\hat{a}$. We then relax the integral over $[\epsilon_{1}, \epsilon_{2}]$ as
\begin{align}
\int_{\epsilon_{1}}^{\epsilon_{2}}\hat{a}\psi(g)dg &<\int_{\epsilon_{1}}^{\epsilon_{2}}\hat{a}\left(\frac{g^{2-\gamma}\tau_t^2}{g^2+\tau_t^2}
+\frac{g^{4-\gamma}}{g^2+\tau_t^2}\right)dg \nonumber\\
&=\frac{R_{net}^2-R_{cell}^2}{R_{cell}^2}\int_{\epsilon_{1}}^{\epsilon_{2}}\frac{a}{R_{net}^2-R_{cell}^2}g^{2-\gamma}dg\nonumber\\
&=(B-1)\mathbb{E}\left[G_{/b}^2\right].
\end{align}
\edf{This indicates that the right hand side of \eqref{eq.se.tinexpr} is larger than that of \eqref{eq.se.recexpr}, and
hence
$\tau_{\infty}^{\mathrm{TIN}}> \tau_{\infty}^{\mathrm{REC}}$.}

\subsection{Proof of Proposition \ref{prop2.1}}
\label{A:propB.1}
\edf{
Using 
\eqref{eq.likeli2} and \eqref{eq.llr},
the
probability of missed detection of the user can be computed as
\begin{align}
P_M^{\mathrm{LS},bn} &=\int_{\mathcal{D}^{\mathrm{LS}}}p(\mathbf{\tilde{x}}_{bbn}|a_{bn}=1)d\mathbf{\tilde{x}}_{bbn}\nonumber\\
&\overset{}{=}\int_{\mathcal{D}^{\mathrm{LS}}} \frac{\exp\left(-\|\mathbf{\tilde{x}}_{bbn}^{t}\|_2^2(\tau_{t}^2+g_{bbn}^2)^{-1}\right)}{\pi^M (\tau_{t}^2+g_{bbn}^2)^{M}}d\mathbf{\tilde{x}}_{bbn},\label{eq.pm.massive.int}
\end{align}
where $\mathcal{D}^{\mathrm{LS}}\triangleq\{\|\mathbf{\tilde{x}}^{t}_{bbn}\|_2^2<l_{bn}^{\mathrm{LS}}\}$.
By noting that
the integral
in \eqref{eq.pm.massive.int} \edww{over}
$\mathcal{D}^{\mathrm{LS}}$ can be interpreted as the CDF of a $\chi^2$ random variable with $2M$ degrees of freedom, 
we can write
\eqref{eq.pm.massive}. Similarly, the probability of false alarm for user $n$ in cell $b$ can be obtained as \eqref{eq.pf.massive}.
}

\subsection{Proof of Proposition \ref{prop5.1}}
\label{A:propD.1}
\edz{
Based on the aggregated LLR expression in \eqref{eq.llr.agg} with a threshold $l_{bn}^{\mathrm{CO}}$ on $ \sum_{j\in\mathcal{B}_{bn}}\Delta_{jbn}\|\mathbf{\tilde{x}}_{jbn}\|_2^2$, the probability of missed detection for user $n$ from cell $b$ is given as
\begin{align}\label{A.eq.pm_coop}
P_M^{\mathrm{CO},bn}= \int_{\mathcal{D}^{\mathrm{CO}}}\prod_{j\in\mathcal{B}_{bn}}p(\mathbf{\tilde{x}}_{jbn}|a_{bn}=1)d\mathbf{\tilde{x}}_{bn}, 
\end{align}
where $\mathcal{D}^{\mathrm{CO}}\triangleq \left\{\sum_{j\in\mathcal{B}_{bn}}\Delta_{jbn}\|\mathbf{\tilde{x}}_{jbn}\|_2^2 < l_{bn}^{\mathrm{CO}}\right\}$. Note that similar to \eqref{eq.likeli2}, when user $n$ from cell $b$ is active, the likelihood of observing $\mathbf{\tilde{x}}_{jbn}^{t}$ at BS $j, j\in \tilde{B}_{bn}$ can be expressed as
\begin{align}\label{A.eq.likelihood}
p(\mathbf{\tilde{x}}_{jbn}^{t}|a_{bn}=1)=\frac{\exp\left(-\|\mathbf{\tilde{x}}_{jbn}^{t}\|_2^2(\tau_{t}^2+g_{jbn}^2)^{-1}\right)}{\pi^M (\tau_{t}^2+g_{jbn}^2)^{M}}.
\end{align}
By plugging \eqref{A.eq.likelihood} into \eqref{A.eq.pm_coop}, we can obtain \eqref{eq.pm_coop}. \edf{The probability of false alarm in \eqref{eq.pf_coop} for user $n$ in cell $b$ can be obtained in a similar way.}
}
\bibliography{chenbib}

\end{document}